\documentclass[11pt]{article}
\usepackage[utf8]{inputenc}

\usepackage{mdframed}
\usepackage{algorithm,algorithm,multicol}
\usepackage{pgfplots}

\usepackage{color-edits}

\addauthor[Alkis]{ak}{niceRed}
\addauthor[P]{pa}{blue}

\definecolor{Gred}{RGB}{219, 50, 54}
\definecolor{Ggreen}{RGB}{60, 186, 84}
\definecolor{Gblue}{RGB}{72, 133, 237}
\definecolor{Gyellow}{RGB}{247, 178, 16}
\definecolor{ToCgreen}{RGB}{0, 128, 0}
\definecolor{myGold}{RGB}{231,141,20}
\definecolor{myBlue}{rgb}{0.19,0.41,.65}
\definecolor{myPurple}{RGB}{175,0,124}
\definecolor{niceRed}{RGB}{153,0,0}
\definecolor{niceRed}{RGB}{190,38,38}
\definecolor{blueGrotto}{HTML}{059DC0}
\definecolor{royalBlue}{HTML}{057DCD}
\definecolor{navyBlueP}{HTML}{0B579C}
\definecolor{limeGreen}{HTML}{81B622}
\definecolor{nicePink}{RGB}{247,83,148}

\usepackage{cmap}
\usepackage[T1]{fontenc}
\usepackage{bm}
\usepackage{amsmath}
\usepackage{amsfonts}
\usepackage{amssymb}
\usepackage{amsbsy}
\usepackage{amsthm}
\usepackage{bbm}

\usepackage{graphicx, ucs}

\usepackage{subcaption}
\usepackage{rotating}
\usepackage{float}
\usepackage{tikz}

\usepackage{algorithm, caption}
\usepackage[noend]{algpseudocode}
\usepackage{listings}

\usepackage{enumitem}

\usepackage[pagebackref]{hyperref}
\hypersetup{
	colorlinks = true,
	urlcolor = {myPurple},
	linkcolor = {royalBlue},
	citecolor = {nicePink}
}
\usepackage[nameinlink]{cleveref}

\crefname{definition}{Definition}{Definitions}

\usepackage{multirow}
\usepackage{array}

\usepackage{chngcntr}

\counterwithin*{equation}{section}
\usepackage{chngcntr}
\usepackage{soul}

\usepackage{nicefrac}

\usepackage{fancyhdr}
\fancyheadoffset{0pt}
\def\compactify{\itemsep=0pt \topsep=0pt \partopsep=0pt \parsep=0pt}
\let\latexusecounter=\usecounter

\definecolor{myC}{rgb}{0, 255, 255}
\definecolor{myY}{rgb}{204, 204, 0}
\definecolor{myM}{rgb}{255, 0, 255}
\definecolor{secinhead}{RGB}{249,196,95}
\definecolor{lgray}{gray}{0.8}

\usepackage{tikz}
\usetikzlibrary{arrows.meta, positioning}

\newtheorem{theorem}{Theorem} 
\newtheorem*{theorem*}{Theorem} 
\newtheorem*{proposition*}{Proposition} 
\newtheorem{lemma}[theorem]{Lemma}
\newtheorem{claim}[theorem]{Claim}
\newtheorem{proposition}[theorem]{Proposition}
\newtheorem{fact}[theorem]{Fact}

\newtheorem{definition}[theorem]{Definition}
\newtheorem{remark}[theorem]{Remark}

\newtheorem{observation}[theorem]{Observation}

\usepackage[most]{tcolorbox}

\newcommand{\poly}{\textnormal{poly}}

\newcommand{\reals}{\mathbb R}

\newcommand{\eps}{\epsilon}
\usepackage{xcolor}

\newcommand{\calY}{\mathcal{Y}}

\def\<{\langle}
\def\>{\rangle}

\makeatletter
\renewenvironment{abstract}{
	\if@twocolumn
	\section*{\abstractname}
	\else 
	\begin{center}
		{\bfseries \large\abstractname\vspace{\z@}}
	\end{center}
	\quotation
	\fi}
{\if@twocolumn\else\endquotation\fi}
\makeatother

\title{CLS Hardness of 2D Min-Max Optimization\\
with Product Constraints}

\author{Alkis Kalavasis\footnote{Work done when AK was supported by the FDS Postdoctoral Fellowship at Yale University.}\\ Harvard University 
\and
Panagiotis Patsilinakos\footnote{PP has received support
under the program ``Investissements d’Avenir'' launched
by the French Government, project PSL JRC 2024 - N 2024 - 488.}\\
Paris Dauphine-PSL University\\
LAMSADE
\and
Stratis Skoulakis\\
Aarhus University
\and 
Manolis Zampetakis\\
Yale University
}

\date{}

\begin{document}

\maketitle
\thispagestyle{empty}

\begin{abstract}
We study the complexity of computing approximate fixed points of gradient descent-ascent (GDA) in min–max optimization problems with one dimensional min and max players and product constraints. For black-box access to the objective, we prove that any deterministic algorithm requires $\Omega(1/\sqrt{\epsilon})$ queries to find an $\epsilon$-approximate GDA fixed point. In the white-box setting, we establish CLS-hardness by reducing from End-of-Potential-Line. These results show that even the simplest min-max setting with one dimensional min and max players is at least as hard as CLS and therefore does not admit polynomial-time algorithms unless $\mathsf{CLS}\subseteq \mathsf{P}$.
\end{abstract}

\section{Introduction}

In this paper we explore the complexity of the computational task of the fundamental problem of computing local equilibria in min-max optimization problems of the form
\begin{equation}
\min_{x \in X} \max_{y \in Y} f(x,y) \label{eq:costis1}
\end{equation}
where $X$ and $Y$ are convex and compact constraint sets, and $f$ is a smooth
objective function. Tracing all the way back to Von Neumann’s celebrated minimax theorem \cite{vonNeumann1928minimax}
and the foundations of Game Theory and Online Learning \cite{Blackwell1956AnAO,10.5555/1137817,Shalev-Shwartz_Ben-David_2014},
the min-max problem has also found multiple applications in Machine learning such as generative adversarial networks models 
\cite{goodfellow2020generative}, 
reinforcement learning \cite{DaskalakisFosterGolowich2020,BaiJin2020,LiWeiChiChen2020}, 
and adversarial robustness \cite{madry2017towards}.
Min-max optimization has also played a central role in the development of other domains such as convex optimization \cite{zbMATH03069631, Adler}.

The most natural solution concept is a pair of strategies $(x^\star,y^\star)$ such that both players (approximately) satisfy the associated first-order Nash equilibrium conditions \cite{lin2020gradient,daskalakis2021complexity,diakonikolas2021efficient}, defined as:
\begin{definition}[First-Order Nash]\label{def:FirstOrderNash}
A pair $(x^\star, y^\star) \in X \times Y$ is an $\epsilon$-first order Nash equilibrium if
\[
(x - x^\star)^\top \nabla_x f(x^\star, y^\star) \geq -\epsilon ~~\mathrm{and}~~
(y - y^\star)^\top \nabla_y f(x^\star, y^\star) \leq \epsilon~~\forall (x,y) \in X \times Y\,.
\]
\end{definition}
It is well-known that, when $f(x,y)$ is a convex-concave function,
i.e., $f$ is convex in $x$ for all $y$ and it is concave in $y$ for
all $x$, then Problem~\eqref{eq:costis1} is guaranteed to have a solution,
under compactness of the constraint set~\cite{vonNeumann1928minimax,rosen1965existence}, while
computing a solution is amenable to convex programming. In fact, if $f$ is
$L$-smooth, the problem can be solved via first-order methods, which are
iterative and only access $f$ through its gradient. When $f$ is non-convex non-concave, Daskalakis, Skoulakis, and Zampetakis \cite{daskalakis2021complexity} proved that there is always a solution (i.e., a first-order Nash equilibrium) to Problem~\eqref{eq:costis1}. They also showed that if the constraints are coupled, i.e., $(x, y) \in S$ where $S$ does not necessarily have the form $S = X \times Y$, then finding such first-order Nash solutions is $\mathsf{PPAD}$-hard. These type of constraints are called coupled constraints and are very crucial in the hardness proof of \cite{daskalakis2021complexity}. Recent papers \cite{bernasconi2024role,anagnostides2025complexity} show a simpler proof of the same result and they clarify and emphasize the importance of these coupled constraints in the hardness results.

It has been further shown in \cite{daskalakis2021complexity} that an equivalent definition of \Cref{def:FirstOrderNash} is based
on approximate fixed points of the gradient descent-ascent map ($\mathsf{GDA}$); namely, \Cref{def:FirstOrderNash} amounts to
bounding the gradient mappings
\[
\max \left\{ \| x^\star - \Pi_X(x^\star - \nabla_x f(x^\star, y^\star))\|, 
\| y^\star - \Pi_Y(y^\star + \nabla_y f(x^\star, y^\star))\| \right\} 
\leq \epsilon'\,,
\]
for some appropriate $\epsilon' = \poly(\epsilon)$, where $\|\cdot \|$ is the $L_2$-norm and $\Pi(\cdot)$
 is the standard projection operator. 
In this paper, we explore the case of computing approximate fixed points of $\mathsf{GDA}$ when $X$ and $Y$ are subsets of the real line $\mathbb{R}$. In this case, the following basic question is still open:
\begin{center}
  \emph{Is there an algorithm for finding approximate $\mathsf{GDA}$ fixed points\\ that is polynomial 
  in $\log(1/\epsilon)$ when $X, Y \subseteq \mathbb{R}$?}\footnote{Observe, that if at least one of the $X$ or $Y$ has dimension at least $2$, then this problem is harder than computing KKT points of a minimization problem, which is known to require $\Omega(\poly(1/\epsilon))$ time \cite{vavasis1993black,fearnley2021complexity}.}
\end{center}

In this work, we answer this question both in the case where the access that we are given to $f$ is \textit{black-box}, i.e., there is an oracle that returns $f(x)$ and $\nabla f(x)$ given $x$, and in the case where we have \textit{white-box} access to $f$, i.e., we are given a Turing Machine that computes $f$ and $\nabla f$. For the black-box model we show that $\Omega(1/\sqrt{\epsilon})$ queries are needed for deterministic algorithms and for the white-box model we show that finding approximate $\mathsf{GDA}$ fixed points is $\mathsf{CLS}$-hard. {This result implies that in the white-box setting, the problem is highly unlikely to admit a polynomial-time algorithm, given our current understanding of computational complexity.}

\paragraph{Complexity classes ($\mathsf{CLS}$).}  For our \emph{white}-box result, we show $\mathsf{CLS}$-hardness. To present $\mathsf{CLS}$ in context, we briefly review and place it in the landscape of total search complexity.

$\mathsf{TFNP}$ is the subclass of total problems in $\mathsf{NP}$ \cite{MEGIDDO1991317}. $\mathsf{TFNP}$ is semantically defined because it is characterized by a global property: every valid instance is guaranteed to have a solution. Thus, it is unlikely to contain ``complete'' problems.\footnote{I.e., there is no unifying computational model do define reductions on.} However, it contains central, syntactically defined classes, such as $\mathsf{PLS}$ and  $\mathsf{PPAD}$. 

$\mathsf{PLS}$ contains local search problems, with polynomial-time computable neighboring improving solution and cost functions. The totality of problems in $\mathsf{PLS}$ is guaranteed by a potential argument, i.e. the fact that any strictly improving local search process in a discrete solution space must eventually terminate at a local optimum (these are problems that can be reduced to finding sink nodes in directed acyclic graphs). It was introduced in \cite{JOHNSON198879} and has been used for example to characterize the complexity of finding pure Nash equilibria in potential games \cite{10.1145/1007352.1007445}. $\mathsf{PPAD}$ was introduced in \cite{PAPADIMITRIOU1994498} to capture a family of problems whose totality is guaranteed by a parity argument (problems in $\mathsf{PPAD}$ can be reduced to finding a source or sink of graphs with in and out degree at most 1). $\mathsf{PPAD}$ has been used to characterize the complexity of problems in continuous domains, such as finding Nash equibria in 2 player games \cite{10.1145/1516512.1516516,doi:10.1137/070699652}. 

$\mathsf{CLS}$ (\emph{Continuous Local Search)} was defined in \cite{10.5555/2133036.2133098} to capture an interesting subclass of problems in  $\mathsf{PPAD\cap PLS}$: optimization problems over a continuous domain, where a potential function is minimized and access to a continuous polynomial-time improvement function is provided.
The breakthrough result of \cite{fearnley2021complexity} showed that $\mathsf{CLS} = \mathsf{PPAD} \cap \mathsf{PLS}$.

\paragraph{Further related work.}
The nonconvex-concave version of Problem~\eqref{eq:costis1} has also been explored \cite{lin2020gradient,ZhangXiaoSunLuo2020,XuZhangXuLan2020}. Here the goal is essentially to find an $\eps$-stationary point of $\Phi(\cdot) = \max_{y \in \calY} f(\cdot, y)$ which is nonconvex. However for any $x$, the maximization problem $\Phi(x)$ can be solved efficiently since $f$ is concave in $y$ and $\calY$ is a bounded convex set. Recently, tractable solution concepts, alternative to Definition \ref{def:FirstOrderNash} have been explored \cite{cai2024tractable}. These solution concepts resemble the notions of correlated equilibrium in classical game theory and they effectively bypass the computational hardness of  \Cref{def:FirstOrderNash}. 

\paragraph{Recent progress in high dimensions.} 
An exciting work by \cite{bernasconi2026complexity} studies the \emph{white}-box variant of the problem for product constraints in high dimensions, showing $\mathsf{PPAD}$-hardness. Their reduction relies on a high-dimensional construction and does not extend to constant dimension (such as $d=2$ which is our focus).
Later on, \cite{bernasconi2026queries} establishes exponential query lower bounds over product hypercubes in growing dimension (where the number of variables/dimensions is allowed to grow with the input size).
\cite{bernasconi2026quadratic} and \cite{anagnostides2026teamzerosum} establish $\mathsf{PPAD}$-hardness for approximate first-order stationary points of quadratic min-max objectives; the latter also establish $\mathsf{PPAD}$-completeness for approximate Nash equilibria in team zero-sum games. 
Although these results restrict the objectives to quadratic
polynomials, their hardness constructions still use a number
of variables that grows with the input size.
Our results concern a broader class of smooth objectives,
but establish $\mathsf{CLS}$-hardness and an
$\Omega(1/\sqrt{\epsilon})$ deterministic query lower bound
already on $[0,1]^2$, with one scalar variable per player. 
Our hardness results concern general smooth objectives;
the objective constructed in our reduction is piecewise
polynomial (obtained by bicubic interpolation), rather than a single polynomial of fixed degree
over the entire domain. 

\paragraph{Follow-up works.} Our CLS hardness result has been applied (via personal communication) to the study of the computational complexity of equilibrium points of electrostatic potentials, i.e., locations in space
where the electrostatic force arising from a collection of charged particles vanishes \cite{ghosh2025computing}.

\section{Preliminaries}
\label{sec:preliminaries}
\noindent In this work, we revisit the problem of computing fixed points of gradient descent-ascent of a function $f$ in two dimensions. Let us first define the dynamics:

\begin{definition}
[Projected Gradient Descent-Ascent]
\label{def:projectGrad}
For a closed and convex $K \subseteq \mathbb{R}^{d_1} \times \mathbb{R}^{d_2}$ and some continuously differentiable function 
$f : K \to \mathbb{R}$, we define the \emph{Unsafe Projected Gradient Descent/Ascent Dynamics with learning rate 1} 
as the map $F_{\mathrm{GDA}} : K \to \mathbb{R}^{d_1} \times \mathbb{R}^{d_2}$ defined as follows
\[
F_{\mathrm{GDA}}(x,y) \triangleq 
\begin{bmatrix}
\Pi_{K(y)} \big( x - \nabla_x f(x,y) \big) \\
\Pi_{K(x)} \big( y + \nabla_y f(x,y) \big)
\end{bmatrix}
\triangleq
\begin{bmatrix}
F_{\mathrm{GDA}x}(x,y) \\
F_{\mathrm{GDA}y}(x,y)
\end{bmatrix}
\]
for all $(x,y) \in K$, where 
\[
K(y) = \{ x' \mid (x',y) \in K \}
\quad \text{and} \quad
K(x) = \{ y' \mid (x,y') \in K \}.
\]
\end{definition}

\begin{remark}
In this work, we are interested in the two-dimensional case and so $d_1 = d_2 = 1$.
\end{remark}
Note that $F_{\mathrm{GDA}}$ is called ``unsafe'' because the projection happens individually for 
$x - \nabla_x f(x,y)$ and $y + \nabla_y f(x,y)$, thus $F_{\mathrm{GDA}}(x,y)$ may not lie in $K$. 
We also define the ``safe'' version $F_{\mathrm{sGDA}}$, which projects the pair 
$(x - \nabla_x f(x,y), \, y + \nabla_y f(x,y))$ jointly onto $K$. 
As shown in \cite{daskalakis2021complexity}, computing fixed points of 
$F_{\mathrm{GDA}}$ and $F_{\mathrm{sGDA}}$ are computationally equivalent so we stick to 
$F_{\mathrm{GDA}}$ which makes the presentation  cleaner.

\paragraph{Lipschitzness and smoothness.} Consider a continuously differentiable function $f:\mathcal{P}\rightarrow \mathbb{R}$, where $\mathcal P\in[0,1]^d$ is some polytope. Then, we say that $f$ is $L$-\emph{Lipschitz} if $|f(x)-f(y)|\leq L\lVert x- y\rVert_2$, for all $x,y$ and $M$-\emph{smooth} if $\lVert\nabla f(x) - \nabla f(y) \rVert_2\leq M \lVert x-y\rVert_2$, for all $x,y$.
\newline

{

For the white-box problem below, rational numbers are represented as
reduced fractions, with numerators and denominators
encoded in binary.
On input a rational point $(x,y)\in[0,1]^2$, the Turing machine $C_f$
returns the exact rational values of $f(x,y)$ and $\nabla f(x,y)$ within
$p(|C_f|+b)$ steps, where $|C_f|$ is its binary description length, $b$ is
the total bit length of the input coordinates, and $p$ is a sufficiently
large polynomial fixed independently of the instance.
We consider the promise version: these outputs agree with an
$L$-Lipschitz and $M$-smooth function defined on the entire real square
$[0,1]^2$.

\begin{tcolorbox}
[colback=gray!20, colframe=black, boxrule=0.5pt, arc=2mm, title=$\mathsf{2D-GDAFixedPoint}$]
\textbf{Input:} Rational parameters $\alpha>0$ and $L,M\geq 0$, and a
Turing machine $C_f$ satisfying the above requirements and evaluating
$f:[0,1]^2\to\mathbb R$ and its gradient
$\nabla f:[0,1]^2\to\mathbb R^2$ at rational points.
\\ \textbf{Output:} A rational point
$(x^\star,y^\star)\in[0,1]^2\cap\mathbb Q^2$ such that
$\|(x^\star,y^\star)-F_{\mathrm{GDA}}(x^\star,y^\star)\|<\alpha$.
\end{tcolorbox}

}

Furthermore, we present here the definition of a classical problem of TFNP \cite{goos2022further,fearnley2020unique}, called End-Of-Potential-Line, which is in spirit similar to End-of-a-Line problem (which characterizes the class $\mathsf{PPAD}$)\footnote{Specifically, $\mathsf{EOPL}$ unifies the problems characterizing $\mathsf{PPAD}$ and $\mathsf{PLS}$: The canonical $\mathsf{PPAD}$-complete problem asks to find an end of a line or a new start, given an exponentially large graph consisting of lines and cycles. The canonical $\mathsf{PLS}$-problem asks to find a potential minimizing vertex in an exponentially large DAG, where the direction of each edge follows the increase of the potential. $\mathsf{EOPL}$ is defined on a graph that only contains lines (or single vertices), with increasing potential along the directed edges (thus it is also a DAG), and asks for either of the solutions accepted by each of the mentioned problems.}, with the additional property that a potential function is given that increases along directed edges. This will be the starting point of our hardness result reduction.

The standard circuit formulation of $\mathsf{EOPL}$ is non-promise.
We use the following equivalent normalized graph formulation, whose
properties are enforced by the local circuit checks described below.

\begin{tcolorbox}[colback=gray!20, colframe=black, boxrule=0.5pt, arc=2mm, title=\textbf{$\mathsf{EOPL}:$ End-of-Potential-Line}]

\textbf{Input:} Access to a directed graph $G = (V, E)$
that is acyclic, has in/out-degree at most 1, and has $|V| = 2^n$ vertices; that is, $G$ is a disjoint
union of directed paths. In addition, we are given the name (binary representation) of a
distinguished {non-isolated} source $v^\star\in V$. The graph is described by the following $\poly(n)$-sized circuits: for any node we
can compute its successor and predecessor, if any, and also an integer potential, which is
guaranteed to increase along any directed path and the potential at $v^\star$ is 0.

\textbf{Output:}
{Any non-isolated source or sink other than $v^\star$,
with in/out-degrees $(0,1)$ or $(1,0)$, respectively.}
\end{tcolorbox}

With some abuse of notation, we use the term $\mathsf{EOPL}$ interchangeably to refer both to the End-of-Potential-Line problem, as well as the complexity class containing all problems that can be reduced to it in polynomial time. The distinction will always be clear by context.

Specifically, if $b(v^*)$ is the binary index of the source vertex $v^*$, we can consider that the graph is described by three polynomial (in $n$) size Boolean  circuits; $\mathsf{Succ, Pred,}$ and $\mathsf{Potential}$. 
\begin{itemize}
\item $\mathsf{Succ, Pred}:\{0,1\}^n\rightarrow \{0,1\}^n$, s.t.: $\mathsf{Succ}(b(v^*))\neq b(v^*) = \mathsf{Pred}(b(v^*))$.  
\item $\mathsf{Potential}: \{0,1\}^n\rightarrow \{0,1, ...,  2^{m-1}\}$, s.t.: $\mathsf{Potential}(b(v^*))=0$.
\end{itemize}

We can consider that $\mathsf{Succ, Pred}$ define the \emph{successor} and \emph{predecessor} functions for each vertex, for an exponentially large (in the size of n) directed, acyclic and max in/out-degree at most 1 graph. $\mathsf{Potential}$ defines the potential function of each vertex, mapping it to an integer. That is, the circuit $\mathsf{Succ}(b(x))$ outputs the encoding of $x$'s (unique - if exists) successor vertex  $y$, i.e., there exists a directed edge $(x,y)$ in $G$ and $\mathsf{Potential}(b(y))>\mathsf{Potential}(b(x))$. If $x$ doesn't have a successor, it outputs $\mathsf{Succ}(b(x)) = b(x)$.

Similarly, $\mathsf{Pred}(b(y)) = b(x)$, if $y$ has a predecessor, $x$, i.e., there is a directed edge $(x,y)$ and $\mathsf{Potential}(b(y))>\mathsf{Potential}(b(x))$. $\mathsf{Pred}(b(y)) = b(y)$ otherwise.

{

To justify this formulation, take the standard EOPL circuits
$S,P,\Phi$ of \cite{fearnley2020unique}, with the locally checkable
conditions $P(s)=s\ne S(s)$ and $\Phi(s)=0$, where $s=b(v^\star)$.
If $P(S(s))\ne s$ or $\Phi(S(s))\leq\Phi(s)$, then $s$ is already
a solution of the standard problem. Otherwise, identifying vertices
with their binary encodings, retain exactly the edges
\[
  (u,v)\in E
  \quad\Longleftrightarrow\quad
  S(u)=v,\quad P(v)=u,\quad \Phi(u)<\Phi(v).
\]
Define $\mathsf{Succ}$ and $\mathsf{Pred}$ to return the corresponding
neighbor in this graph, or the input vertex itself when no such neighbor
exists, and set $\mathsf{Potential}=\Phi$.
These polynomial-size circuits enforce consistency, in/out-degree at
most one, and acyclicity, while preserving the distinguished source.
Every non-isolated sink $v$ is a standard EOPL solution. For any other
non-isolated source $v$, if $S(P(v))\ne v$, then $v$ is a standard
solution; otherwise $\Phi(v)\leq\Phi(P(v))$, so $P(v)$ witnesses a
failure of strict potential increase and is a standard solution.
Conversely, the consistent circuits of a graph in the box above form
a standard EOPL instance whose solutions are exactly the permitted
endpoints. Thus this graph formulation is polynomial-time equivalent
to standard EOPL, without requiring a global promise on the input
circuits.

}

The fundamental results of \cite{goos2022further} and \cite{fearnley2021complexity} show us that $\mathsf{EOPL}$ characterizes $\mathsf{CLS}$:
\begin{proposition*}
[See \cite{goos2022further,fearnley2021complexity}]
$\mathsf{EOPL} = \mathsf{PLS} \cap \mathsf{PPAD} = \mathsf{CLS}.$
\end{proposition*}

\section{Main Contributions}
Our first contribution is to  establish the  $\mathsf{CLS}$-hardness of $\mathsf{2D-GDAFixedPoint}$. This result combines the recently developed techniques for the complexity of KKT points \cite{fearnley2021complexity} together with the recent characterization of the class $\mathsf{EOPL}$ (i.e., \cite{goos2022further}).
\begin{theorem}
$\mathsf{2D-GDAFixedPoint}$ is $\mathsf{CLS}$-hard.
\end{theorem}

For the proof, we refer to \Cref{sec:cls}. Let us briefly compare our technique with prior work. Our work diverges from the techniques of \cite{daskalakis2021complexity} since they crucially rely on non-product constraints and instead follows the approach of \cite{fearnley2021complexity} but with some important deviations.

{
The construction of \cite{fearnley2021complexity} concerns
approximate KKT points for minimization, whereas our construction
concerns approximate fixed points of GDA. For minimization of
$f$ over a box $K\subseteq\mathbb R^2$, a
\emph{Karush--Kuhn--Tucker (KKT) point} is a point
$z^\star\in K$ satisfying
\[
  z^\star=\Pi_K\bigl(z^\star-\nabla f(z^\star)\bigr).
\]
Exact KKT points and exact GDA fixed points coincide in the
interior, where both conditions are equivalent to
$\nabla f(z^\star)=0$, but they can differ at the boundary
because GDA uses ascent in the $y$-coordinate.

From this perspective, we highlight two important differences between our construction
and analysis and those of \cite{fearnley2021complexity}.
First, we introduce new \emph{path gadgets}, namely local
pieces of the function that encode paths in the EOPL graph,
adapted to gradient descent in $x$ and gradient ascent in $y$.
Second, although we reuse gradient estimates from
\cite{fearnley2021complexity}, these do not by themselves
exclude unwanted approximate GDA fixed points where projection
is active. We therefore provide a separate analysis near the
boundary, ensuring that every approximate GDA fixed point
of the required accuracy lies in a region encoding an EOPL
solution; see \Cref{sec:valid} and \Cref{app:cases}.
}

Our second main contribution focuses on the \textit{black-box} perspective and provides a lower bound on the number of required queries to the function and its gradient, in order to reach an approximate gradient descent-ascent fixed point. 

\begin{theorem}
[Query Lower Bound]
\label{thm:2}
For any $N > 145$ and any deterministic algorithm $A$ with access to function and gradient evaluation oracles, there is function $f \equiv f_A : [0,N]^2 \to \mathbb R$ that is $100 N$-smooth for which $A$ has to make, in the worst case, $\Omega(N)$ queries in order to find an $0.01$-approximate fixed point of GDA.
\end{theorem}

For the proof of this result, we refer to \Cref{sec:query}.
{For convenience, our constructions are described on $[0,N]^2$.
The following rescaling transfers both the query lower bound and the
CLS-hardness reduction to the domain $[0,1]^2$ used in the formal problem
definition, while keeping the learning rate in the GDA map equal to~$1$.}

\begin{remark}
\label{rmk:scaling}

{Let $N\geq 1$, and let $f:[0,N]^2\to\mathbb{R}$ be continuously
differentiable, $L$-Lipschitz, and $M$-smooth. Define
\[
  \widetilde f(z)=\frac{f(Nz)}{N^3},
  \qquad z\in[0,1]^2.
\]
Then
\[
  \nabla\widetilde f(z)=\frac{\nabla f(Nz)}{N^2}.
\]
Moreover, $\widetilde f$ is $\widetilde L$-Lipschitz and
$\widetilde M$-smooth, where
\[
  \widetilde L=\frac{L}{N^2},
  \qquad \widetilde M=\frac{M}{N}.
\]

To verify the effect of the transformation on approximate fixed points, let
$F_{\mathrm{GDA},h}$ denote the GDA map of \Cref{def:projectGrad} for objective
$h$ on its stated domain, with learning rate~$1$, and write
$R_f(w)=\|w-F_{\mathrm{GDA},f}(w)\|_2$ and
$R_{\widetilde f}(z)=\|z-F_{\mathrm{GDA},\widetilde f}(z)\|_2$.
Fix $z\in[0,1]^2$, let $w=Nz$, and set
\[
  v(w)=\bigl(-\nabla_x f(w),\,\nabla_y f(w)\bigr).
\]
For each coordinate $i\in\{1,2\}$, let $a_i=|v_i(w)|$ and let $d_i$
be the distance from $w_i$ to the boundary of $[0,N]$ in the direction
of $v_i(w)$. Thus, $d_i=N-w_i$ if $v_i(w)>0$, and $d_i=w_i$ if
$v_i(w)<0$; when $v_i(w)=0$, set $d_i=0$.
Projection onto an interval truncates the proposed displacement at that
boundary. Consequently,
\[
  \bigl|\bigl(F_{\mathrm{GDA},f}(w)-w\bigr)_i\bigr|
     =\min\{a_i,d_i\},
  \qquad
  \bigl|\bigl(F_{\mathrm{GDA},\widetilde f}(z)-z\bigr)_i\bigr|
     =\min\left\{\frac{a_i}{N^2},\frac{d_i}{N}\right\}.
\]
This implies
\[
  \frac{\min\{a_i,d_i\}}{N^2}
  \leq \min\left\{\frac{a_i}{N^2},\frac{d_i}{N}\right\}
  \leq \frac{\min\{a_i,d_i\}}{N}.
\]
Taking Euclidean norms gives
\[
  \frac{R_f(Nz)}{N^2}
  \leq R_{\widetilde f}(z)
  \leq \frac{R_f(Nz)}{N}.
\]
In particular, for every $\varepsilon_0>0$,
\[
  R_{\widetilde f}(z)<\frac{\varepsilon_0}{N^2}
  \quad\Longrightarrow\quad
  R_f(Nz)<\varepsilon_0.
\]
}
\end{remark}

{For the instances in \Cref{thm:2}, $M=100N$ and
$\varepsilon_0=0.01$. Hence $\widetilde f$ is $100$-smooth, and any
$(0.01/N^2)$-approximate GDA fixed point of $\widetilde f$ yields a
$0.01$-approximate GDA fixed point of $f$. Each function or gradient query
to $\widetilde f$ can be simulated by one corresponding query to $f$
at $Nz$, followed by division by $N^3$ or $N^2$, respectively.
The $\Omega(N)$ query lower bound therefore transfers to the unit square.
More explicitly, for any sufficiently small target accuracy $\varepsilon$,
choose $N=\lfloor\sqrt{0.01/\varepsilon}\rfloor$. Then
$\varepsilon\leq0.01/N^2$ and $N=\Theta(\varepsilon^{-1/2})$.
Thus, for every deterministic algorithm with function and gradient oracle
access, there is an $O(1)$-smooth function on $[0,1]^2$ for which finding
an $\varepsilon$-approximate GDA fixed point requires
$\Omega(\varepsilon^{-1/2})$ queries in the worst case.

For the remainder of the construction and its correctness analysis, we
work on $[0,N]^2$ for convenience. In \Cref{subsec:ReductionComplete} we explicitly apply
\Cref{rmk:scaling} to obtain the unit-square instance required by the formal
CLS-hardness reduction. We proceed with an outline of our technical tools.
}

\begin{remark}
[First-order Nash points and GDA fixed points]
{
An $\epsilon$-first-order Nash point is not necessarily a $\epsilon$-approximate GDA fixed point: passing from one notion to the other may require a different approximation accuracy. Every $\epsilon$-first-order Nash point $z$, satisfies:
\[
\|z - F_{GDA,f}(z)\|_2\leq \sqrt{2\epsilon}.
\]
Indeed, for either coordinate, let $a$ be the absolute value of the (corresponding partial) derivative and $d$ the distance to the boundary for the improving direction. \Cref{def:FirstOrderNash} gives $ad\leq \epsilon$, while the coordinate changes by $\min\{a,d\}\leq \sqrt{\epsilon}$.

Under the normalization $\widetilde f(z)=f(Nz)/N^3$, the expressions in \Cref{def:FirstOrderNash} are divided by $N^3$. Therefore, choosing $\epsilon_0=0.01$, any $\epsilon_0^2/(4N^3)$-first-order Nash point $z$ of $\widetilde f$ satisfies
\[
\|Nz-F_{GDA,f}(Nz)\|_2\leq \epsilon_0/\sqrt{2}<\epsilon_0.
\]

Thus, the same reduction proves CLS-hardness for first-order Nash points of \cref{def:FirstOrderNash} with constant smoothness. An $\Omega(N)$ query lower bound for the original construction on $[0,N]^2$ gives a lower bound of $\Omega(\epsilon^{-1/3})$ queries for Nash points under this normalization, while it gives $\Omega(\epsilon^{-1/2})$ for approximate GDA fixed points. Alternatively, dividing by $N^2$ instead of $N^3$ gives first-order Nash accuracy of $\epsilon_0^2/(4N^2)$ and thus a query lower bound of $\Omega(\epsilon^{-1/2})$, while for GDA it gives accuracy $\epsilon_0/N$ and query lower bound of $\Omega(\epsilon^{-1})$; however the smoothness becomes $O(N)$ instead of constant.\footnote{Notice that both normalizations preserve the CLS-hardness results, and only change the corresponding choice of approximation accuracy.} 
}
\end{remark}

\section{Proof Outline}

Our proof technique is inspired and follows the breakthrough result of \cite{fearnley2021complexity} combined with recent progress by \cite{goos2022further} on the complexity of $\mathsf{EOPL}$.
$\mathsf{EOPL}$ is the End-Of-Potential-Line problem, which was introduced by \cite{fearnley2020unique}.

The key technical ingredient is to construct a class of ``hard functions'' that can be used to encode End-Of-Potential-Line instances.  The reduction consists of two parts: a geometric part, that demonstrates the shape of the functions in the class, and a complexity theoretic part, that shows a way to embed $\mathsf{EOPL}$ as functions in our hard class. The main novelty of our work lies on the geometric part. The techniques of the complexity theoretic part require various details but mostly follow the blueprint of prior works \cite{fearnley2021complexity,hollender2023computational}. 

In this section, we will focus on a single monotone path. Such an instance is not sufficient to embed EOPL (since EOPL may contain multiple disjoint paths) but it gives a rough idea of our geometric construction. The embedding of EOPL is presented in the next section.

Our constraint set will be the square $[0,N]^2$ for some large $N$. We will first work on the discrete grid $G = \{0,1,...,N\}^2$, which is the discretization of $[0,N]^2$, define our collection of hard functions on $G$ and then extend each one of them to the continuous domain.

The rough idea is that the $N \times N$ grid $\{0,N\}^2$ will be divided into sub-squares, each sub-square having size $c \times c$, for some absolute constant $c$ ($c$ will be 6 in our constructions). We label the sub-squares, which we call \emph{tiles}, with the ordered pairs $\<u,v\>$ where $u,v = 0,1,...,N-1$.

The main ingredient is to understand how to construct functions in two dimensions that are associated to \emph{monotone paths}:
\begin{enumerate}
    \item Each function in the class will have Lipschitz constant $L$ and smoothness $M$ which satisfy $L, M \gtrsim N$.

    \item  Each function of the ``hard family'' will resemble a ``monotone'' path, starting from the bottom left sub-square and moving either up or right (to adjacent sub-squares) until reaching the \emph{sink} sub-square. The idea is that the algorithm will have to ``chase'' this path until it reaches its end. See \Cref{fig:pathAndTiles} for an illustration.

    \item The sink sub-square of this monotone path will be the only place that an approximate fixed point of $\mathsf{GDA}$ will exist (with constant accuracy).
\end{enumerate}

\begin{figure}[!ht]
    \centering
    \includegraphics[width=0.7\linewidth]{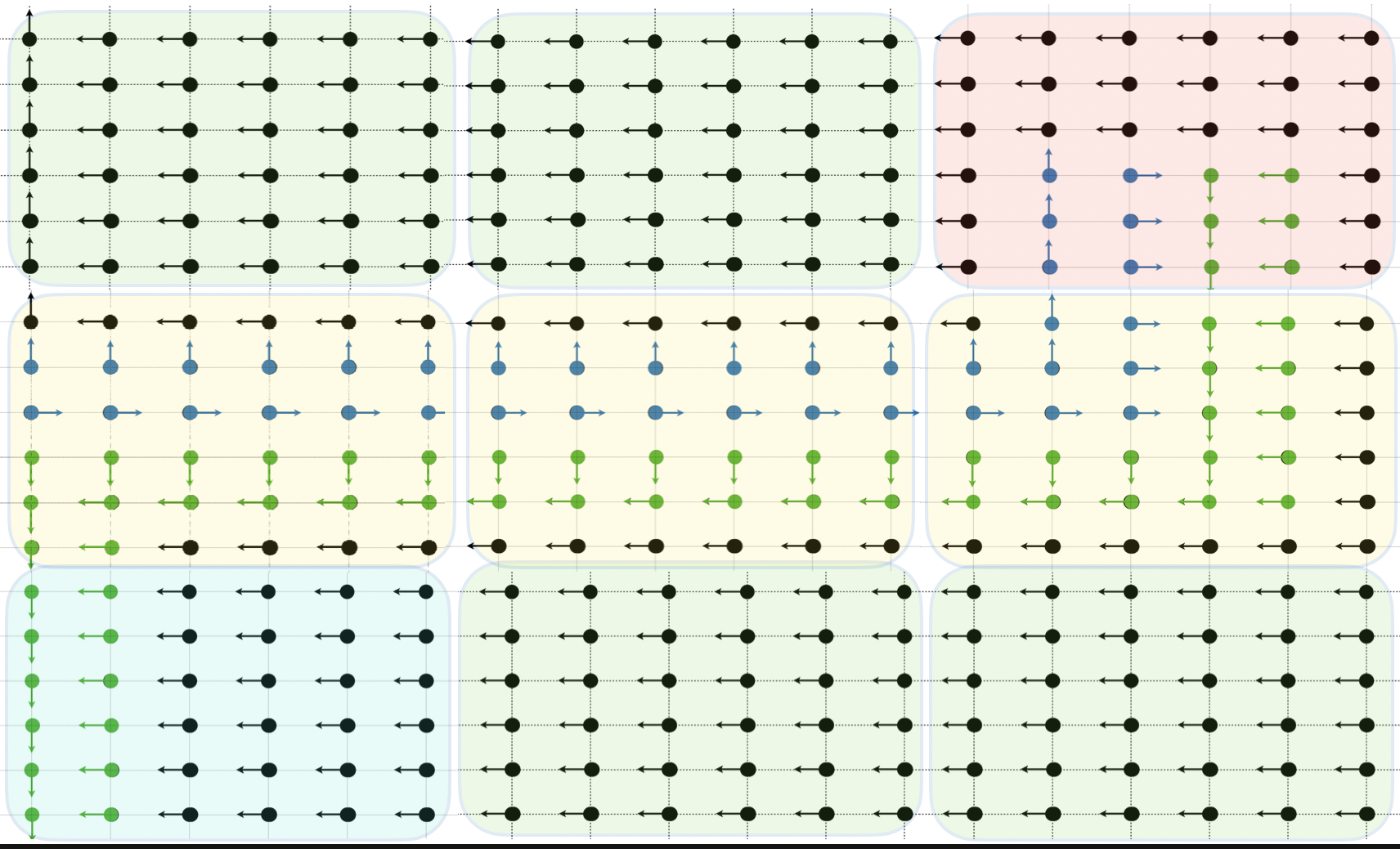}
    \caption{Discretization of square (here we have a 18 $\times$ 18 array of grid points and each tile is 6 $\times 6$) and a monotone path. The path starts from the bottom left tile and keeps moving on the yellow tiles, moving right and finally up until reaching the red tile which contains a sink. The green tiles complement the grid and are the ``background'' - they contain only black colored points that are not part of the monotone path (``environment'' points). The monotone path consists of the points that are colored green and blue and except for the source tile, it has width equal to 4 (2 blue and 2 green points). The \textit{min} player wants to ``move'' parallel to the $x$ axis towards the direction of the arrows (function decreases) and the \textit{max} player wants to ``move'' parallel to the $y$ axis, on the opposite direction pointed by the arrows (function increases).}
    \label{fig:pathAndTiles}
\end{figure}

 Given the above general plan, we have to explain how these monotone paths are constructed. The challenging part is to ensure that the only solutions, i.e., fixed points of $\mathsf{GDA}$, appear in the end of the monotone path (sink). 

\paragraph{Constructing Monotone Paths} Now, we will illustrate how the monotone path of \Cref{fig:pathAndTiles} (and hence an associated function) can be constructed. The family will then contain a collection of such paths. We can decompose the construction as follows. First, we will fix a monotone path, i.e., a sequence of adjacent sub-squares moving from the source either up or right, until reaching a sink. This monotone path will induce a function, as we will describe later. We will initially define the function on the points of the grid $G = \{0,1,...,N\}^2$, and then extend it to the continuous $[0,N]^2$ using bicubic interpolation \cite{fearnley2021complexity}. Our construction consists of the following steps.

\begin{enumerate}
    \item (Defining the function values) In the first step, we have to assign function values to each point of the grid $G$, i.e., we have to determine the values $f(x,y)$ for any $(x,y) \in G$. For that, we introduce 3 colors, named 0 (black), 1 (blue) and 2 (green), that map $(x,y)$ to a value: for each point $(x,y)$, we will assign a color (and hence a function value). The way that we color each point will be made clear in the next step. The details appear in \Cref{sec:function-values}.

    \item (Coloring the points) Now that we have introduced the three potential colors, we will color the grid in order to form the monotone path that we discussed before. (Recall that for the construction of a function in our class, we have fixed a monotone path.) This monotone path can be thought as a composition of small path gadgets. Since the path is moving only up and right, we will need four path gadgets: one for moving right, one for moving up, one for creating the up-right turn, and one for creating the right-up turn. Also, we will need a gadget for the source of the path and one for the sink. In total, we introduce 6 path gadgets. The path gadgets will be used to form the whole path in the sub-squares of the grid, i.e., each sub-square of $[0,N]^2$ will be filled with one of the path gadgets. Each gadget is first defined on the grid and contains a carefully chosen coloring (which gives the function values on the points of the grid) and gradient directions (indicating the function value decay on the points of the grid). The construction of these gadgets is one of the main contributions of our work and appears in \Cref{sec:path-gadgets}.

    \item (Extending the function)
    Since we have fixed the monotone path, defined as a sequence of sub-squares as discussed before, we use the appropriate path gadgets to define the function on the part of the grid that is traversed by the path. The remaining part of the grid is called environment and is simple to be created; roughly speaking, it contains only color 0 and the function value is decreasing in the direction of one of the boundaries (depending on whether the sub-square is above or below the monotone path). Given these, we have defined our function on the whole grid $G$. We now have to extend our function to $[0,N]^2$.
    The key idea is to use \emph{bicubic interpolation}, a classical technique for  interpolating data points on a two-dimensional regular grid. This part is based on the techniques of \cite{fearnley2021complexity}. We discuss on this extension in \Cref{sec:extension}.

    \item (Avoiding Spurious Equilibria) The most crucial aspect of our construction is to show that approximate fixed points of $\mathsf{GDA}$ only lie in the sink path gadget. This will follow from the way we constructed the path gadgets and the techniques introduced in \cite{fearnley2021complexity}.
    Crucially, the work of \cite{fearnley2021complexity} does not directly give us the desired property. The reason is that, {when the projection operators are not active}, then fixed points of $\mathsf{GDA}$ and $\mathsf{KKT}$ points coincide. {In that regime, the work of \cite{fearnley2021complexity} guarantees that there are no solutions outside of the sink-subsquares.} {However, this does not suffice for our purposes, since when at least one of the projection steps are activated, then it is unclear whether a fixed point exists outside of the sink. To this end, we have to give a new analysis for this case that relies explicitly } on the structure of our path gadgets.
    We prove our key property in \Cref{sec:valid}.
    It is helpful to remark that this property is, in parallel, with the seminal lower bounds of \cite{hirsch1989exponential,vavasis1993black}. For instance, Vavasis constructed (explicitly) a collection of functions with the property that any local minimum is in the sink sub-square. This was used to prove that local minimization required at least $\sqrt{M/\epsilon}$ queries in two dimensions. 
\end{enumerate}

\paragraph{The Key Technical Challenges} There are two key technical challenges that one needs to bypass in order to make the above plan rigorous. First, the functions constructed by \cite{fearnley2021complexity} are not sufficient for our purposes since we care about fixed points of GDA, which does not coincide with 2-dimensional KKT points. To this end, we have to construct our new path gadgets, which leads us to a new class of hard instances. The second issue is more important and comes from the fact that KKT points are not the same as GDA points. Hence, the techniques of \cite{fearnley2021complexity} do not directly imply that our hard instances do not have ``bad'' GDA solutions. To this end, we first observe that GDA fixed points and 2-dimensional KKT points coincide when none of the GDA projections are activated. We call these points \emph{good}. For good points, we can show using the tools developed in \cite{fearnley2021complexity} that the constructed instances do not have ``bad'' fixed points, i.e., fixed points that do not belong at a sink sub-square. However, for the points that are not good, we have to make a separate analysis. This requires some technical effort since we have to analyze all the possible gadgets that could have bad points and we have to show that in each one of them, bicubic interpolation does not introduce GDA fixed points (for details, we refer to \Cref{app:cases}).

\subsection{Query Lower Bound}
While the above collection of monotone paths (and hence functions) cannot be immediately used to prove the desired CLS hardness result, we can almost immediately use this class of hard functions, to show a query lower bounds for deterministic algorithms for computing approximate fixed points of gradient descent-ascent. For the statement and the details, we refer to \Cref{sec:query}.

\section[CLS hardness of 2D Min-Max with Product Constraints]{$\mathsf{CLS}$ hardness of 2D Min-Max with Product Constraints}
\label{sec:cls}

For this proof, we will use the  results of \cite{goos2022further} and \cite{fearnley2021complexity}, as mentioned in \Cref{sec:preliminaries}:
\begin{proposition}
[See \cite{goos2022further,fearnley2021complexity}]
$\mathsf{EOPL} = \mathsf{PLS} \cap \mathsf{PPAD} = \mathsf{CLS}.$
\end{proposition}

This shows a collapse of the class EOPL with CLS. This allows us to use EOPL as the starting point of our CLS-hardness proof. The proof consists of two steps:
\begin{enumerate}[label=\Roman*]
    \item As a first step, we will see how to embed instances of EOPL as \emph{monotone paths} to a two-dimensional grid of size $[0,N] \times [0,N]$. These paths need to satisfy some important properties as we will see below, such as they should be monotone, in the sense that they only move up and right and should not cross. 

    \item Given such monotone paths, we will see how to associate with each EOPL instance  a \emph{function} defined on $[0,N] \times [0,N]$ whose approximate GDA fixed points (for some level of accuracy) coincide exactly with the set of solutions of the EOPL instance (i.e., its sinks and the sources which do not coincide with the given one). The values of this function will decrease as we move towards each monotone path.
\end{enumerate}

This means that any algorithm that could efficiently find GDA fixed points for these functions, would also be able to solve EOPL. To give some intuition for readers familiar with PPAD and PLS, our constructed functions have two key properties; (i) they can be equivalently seen as paths over a square grid (\textsc{PPAD}-component) and (ii) as the monotone path moves up and right, the function values drop and so we can think that one minimizes a potential function ($\textsc{PLS}$-component). 

Hence, the main result of this section is the following:  
\begin{theorem}
\label{thm:mainCla}
    $\mathsf{EOPL} \le_P \mathsf{2d}-\mathsf{GDAFixedPoint} $. 
\end{theorem}

\begin{proof}

{
Given an instance $I^{\mathrm{EOPL}}$ with $2^n$ nodes and circuits
$\mathsf{Succ, Pred}$ and $\mathsf{Potential}$, we first construct
an auxiliary function $f:[0,N]^2\to\mathbb{R}$, where $N$ is the
integer side length of the constructed grid and {$\log N=O(n+m)$, where $m$ is the number of bits encoding the potential value},
such that every $0.01$-approximate GDA
fixed point of $f$ can be efficiently converted into a solution of
$I^{\mathrm{EOPL}}$. The construction provides polynomial-time
evaluators for $f$ and its gradient. The actual instance of
$\textsc{2D-GDAFixedPoint}$ produced by the reduction is the normalized
function $\widetilde f(z)=f(Nz)/N^3$ on $[0,1]^2$, with accuracy
$\alpha=0.01/N^2$, as justified by \Cref{rmk:scaling} and detailed in \Cref{subsec:ReductionComplete}.
}

This result combined with the fact that $\mathsf{EOPL} = \mathsf{CLS}$ implies our main result for the white-box model. For convenience, the main steps of the proof are presented in separate subsections.

\subsection{Step I: Embedding EOPL in 2D Grid}

For the first part of the reduction, we use the EOPL instance from
\Cref{sec:preliminaries}. Let us recall the instance: its directed graph $G=(V,E)$ is a disjoint union of
directed paths on $2^n$ vertices, described by the circuits
$\mathsf{Succ}$, $\mathsf{Pred}$ and $\mathsf{Potential}$. The potential
increases along every edge of the path (hence the path is monotone), and $v^\star$ is the distinguished source.
Let $m$ be the number of output bits of $\mathsf{Potential}$, and write
$\Phi(v)=\mathsf{Potential}(v)$ and
$\operatorname{id}(v)\in\{0,\ldots,2^n-1\}$ for the integer encoded by the
label of a vertex $v$.

Our goal is to embed the paths of the instance in a two-dimensional grid and then use
the gadgets of the upcoming \Cref{sec:grids} to define a function on $[0,N]^2$.
We note that
the embedding must be {locally computable}. We start with some terminology.

\begin{definition}[Grid points, tiles and coarse blocks]
The grid points are the integer-coordinate points of $[0,N]^2$.
A unit square is $[a,a+1]\times[b,b+1]$ for integers
$0\leq a,b<N$. A tile contains a $6\times6$ array of grid points and
carries one of our basic gadgets. A coarse block groups $3\times3$ tiles,
and hence contains an $18\times18$ array of grid points.
Unit squares also occur between neighboring tiles and coarse blocks.
\end{definition}

\paragraph{Grid and numerical choices.}
We use $B$ coarse blocks along each axis, indexed by
$\{0,\ldots,B-1\}$, where
\[
 B=2^{n+m}+4.
\]
Our goal is to position the input vertices (which come with a label with $n$-bits and a potential value with $m$-bits) into the \emph{diagonal} of the grid. Hence,
the $2^{n+m}$ term provides one possible diagonal position for each
pair consisting of an $m$-bit potential and an $n$-bit vertex label.
The extra four positions will be used of some padding operation, as we will explain below.
Then the number of tiles $T = 3B$ comes from the definition of a tile and the number of grid points is $N+1 = 18B$. In particular, $N$ denotes the side length of
the domain.

\paragraph{Placing the vertices and padding.}
We can then place vertex $v$ in the diagonal coarse block
$\langle\rho(v),\rho(v)\rangle$, where the \emph{rank} $\rho$ is defined as
\[
 \rho(v)=2+2^n\Phi(v)+\operatorname{id}(v).
\]
Before adding the $2$, the rank of $v$ is the concatenation of the potential
bits followed by the vertex-label bits. Note that if $u\to v$, then $\rho(u)<\rho(v)$ and a vertex's position is obtained
with one evaluation of $\mathsf{Potential}$.

We artificially added the term 2 and we added the term 4 in the definition $B$ in order to ensure that the first two and
last two coarse rows and columns contain no vertices. This is what we refer to as
padding. In this way, ordinary vertices and their connecting paths stay away from
the physical boundary. We will use this extra no vertex space to place a special source node in the final construction.

\paragraph{Connecting the edges.}
To embed an edge $u\to v$ in the 2D grid, we start from the encoding 
$\langle\rho(u),\rho(u)\rangle$ and first move horizontally to
$\langle\rho(v),\rho(u)\rangle$, and then vertically upward to
$\langle\rho(v),\rho(v)\rangle$. These paths are monotone, in the sense that
they traverse only right and up. Note that the order is important: outgoing edges
leave to the right, and incoming edges arrive from below, matching
the entrances and exits of our gadgets. 
\Cref{fig:path-creation} illustrates the idea of ordering vertices by
potential and label. Crossings are handled later by some standard local re-routing, as
described in \Cref{sec:pathCrossings}.

\begin{figure}[!ht]
    \centering
{\includegraphics[width=0.7\linewidth]{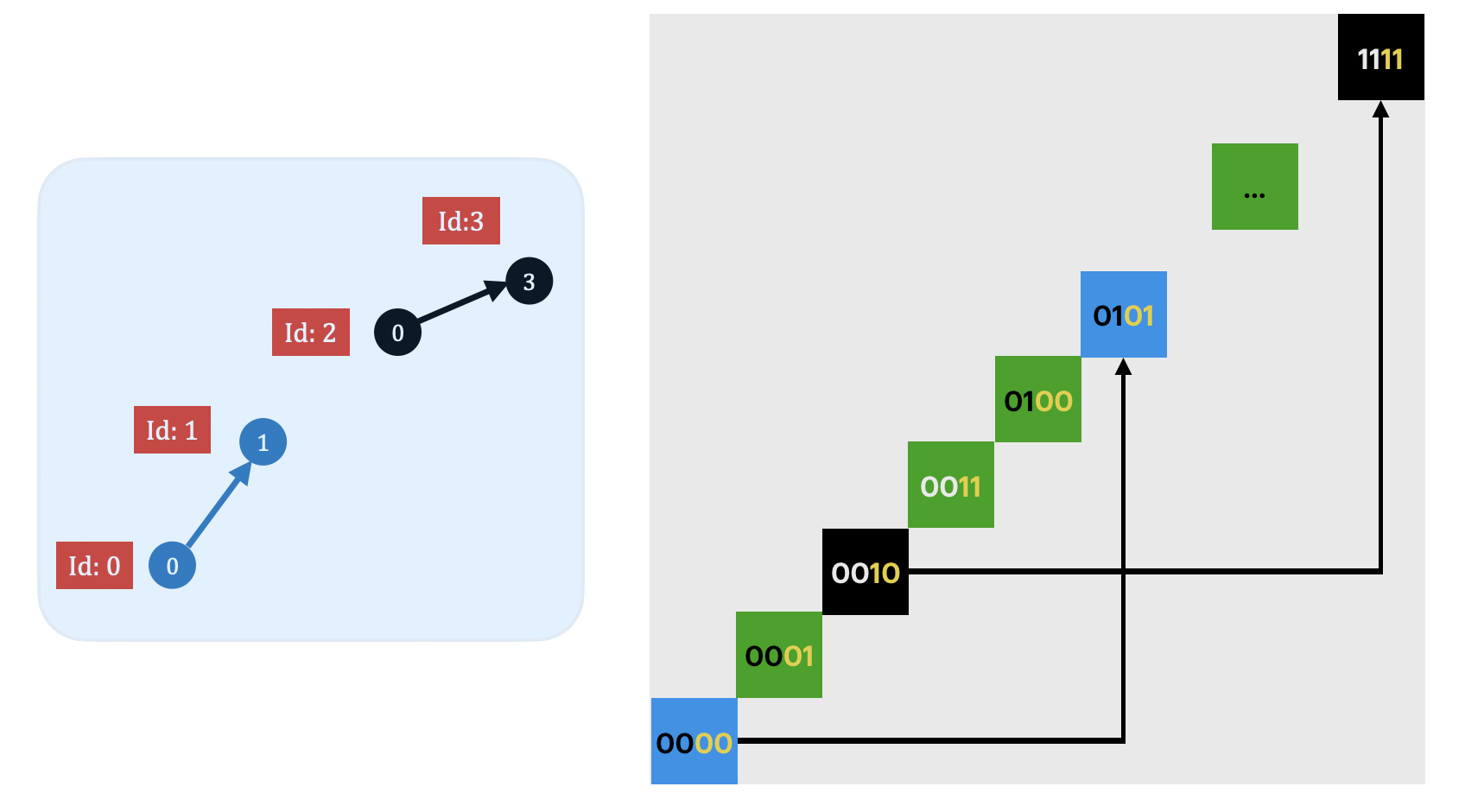}}
    \caption{{Schematic illustration of ordering vertices by potential and label,
before padding. The actual vertex
positions are given by $\rho(v)$. In the right-hand
diagram, each edge first goes right and then up, leaving its source
to the right and entering its destination from below. Crossings are
subsequently replaced by the local block gadgets.}}
    \label{fig:path-creation}
\end{figure}

\paragraph{Grid details.}
Given the above embedding, we can classify our sub-squares in the following categories. Recall
that eventually our goal is to define a function over the grid.
\begin{enumerate}
 \item \emph{Background.} The blocks which are not used by a path will receive the background
 gadgets. We will use a simple function to define those. In particular, for $(x,y)$ not in the path, we use the function
$f(x,y) = x-y+4N+10$ (the reason will become clear in the next section). Observe that if
there are no paths, the only solution appears in the origin.
 
 \item \emph{The source.} The origin, i.e., the bottom left part of the grid is very important. Since it is not a solution of
EOPL, we have to construct our function in a way that it does not have a GDA fixed point
there. We will see how to handle this later but we mention that we will use a special gadget moving from $\<0,0\>$ to $\<\rho(v^\star),0\>$ and then to $\<\rho(v^\star),\rho(v^\star)\>$.
 
 \item \emph{Other path endpoints.} A path can start in the background
 at a non-trivial source or end in the background at a sink. These
 are valid EOPL solutions, so the corresponding endpoint regions are
 precisely where our construction permits approximate GDA fixed points.
\end{enumerate}

Every vertex $v\neq v^\star$ has at most one incoming and one outgoing
edge. Accordingly, its diagonal block has one of the following four forms:
\begin{enumerate}
 \item $v$ is isolated: Then the block will not contain any path and will fully be in the
background, thus not containing any GDA fixed point.

 \item $v$ has one incoming edge and no outgoing edge: the block will contain
the end of a monotone path. There will be a GDA fixed point at the end of the path, which
is fine, since $v$ is a sink of the End-of-Potential-Line instance.

 \item $v$ has one outgoing edge and no incoming edge: a non-trivial-source
 gadget starts a path to the right. There will be a GDA fixed point at the start of the path, which
is fine, since $v$ is a non-trivial source of the End-of-Potential-Line instance.
 
 \item $v$ has both an incoming and an outgoing edge: in this case, we will be able to connect the two
paths without introducing GDA fixed points, i.e., implementing a turn.
\end{enumerate}

\subsection{Step II: From Monotone Paths to Functions on Grids}
\label{sec:grids}

Now we have a collection of monotone paths placed in our $N \times N$ grid and can successfully embed any EOPL instance with $2^n $ nodes.
This section now explains how to define some particular path gadgets in order to define our function $f$, which will be used to define the instance of $\mathsf{2d-GDAFixedPoint}$.

For our purposes, it will be more convenient to work with a discrete grid $G.$
Let us consider the grid $G = \{0,1,...,N\} \times \{0,1,...,N\}$. Our goal is to color and assign a direction to each point of the grid $G$. The colors will correspond to function values and the directions to gradients of a function, which will be later obtained as the continuous extension of this discrete collection of values and gradients to the continuous $[0,N]^2$ cube (using bicubic interpolation in two dimensions).

\subsubsection{Defining the function values}
\label{sec:function-values}
In this section, we introduce the way that we will assign function values to the members of our class. We will use three colors, 0 (black), 1(blue) and 2 (green). Let us now specify the (local) function values. In particular, if a grid point $(x,y) \in G$ is:
\begin{enumerate}
    \item in the color 0 regime, we set $f(x,y) = x-y + 4N + 10$
    \item in the color 1 regime, we set $f(x,y) = -x-y + 8N + 20$
    \item in the color 2 regime, we set $f(x,y) = x+y$
\end{enumerate}

Some remarks are in order. One can check that color 1's values are strictly larger than color 0's values, which are strictly larger than color 2's values and that
the gap between any two color regimes at any point is at least 10. Note that $3N+10 \leq \mathrm{color~0~values} \leq 5N + 10$,
$6N+20 \leq \mathrm{color~1~values} \leq 8N + 20$, and,
$0\leq \mathrm{color~2~values} \leq 2N$.

The black color regime will be used for the ``background'', i.e., anything that is not part of the path. 
We will re-scale the function at the end of the construction to deal with the domain $[0,1]^2$.

\subsubsection{Path gadgets, coloring of the grid, and monotone paths}
\label{sec:path-gadgets}
Given a grid $G$ of size $N \times N$, any function $f$ on it will correspond to a coloring of each $(x,y) \in G$ (function value) together with a direction, indicating the descent in function value. 
To design each hard function in our family, we introduce the following steepest-descent path gadgets for our path.
All the path gadgets correspond to a $6 \times 6$ sub-square of the grid $G$. Hence, a function will be a result of padding many of those path gadgets in our $N \times N$ grid (with some constraints). 

The main constraint is that each function in the hard family will intuitively correspond to a collection of ``monotone'' paths, i.e., paths that move either up or right. The construction specifies for each grid point a color (which represents the value of $f$ at
that point) and an arrow that represents the direction of $-\nabla f$ at that point. A general
rule that one should follow throughout the construction is that the function values should be
consistent with the arrows. For example, if some grid point has an arrow pointing to the
right, then the adjacent grid point to the right should have a lower function value, while
the adjacent grid point to the left should have a higher function value.
A key property of the functions in our class is that there will only be an approximate fixed point of $\mathsf{GDA}$ at the \emph{sink} of the associated monotone path. We will prove this later.

\subsubsection{Path Formulation}

We proceed with a formal description of our gadgets.

\begin{definition}
    [Tile]
    Given a grid $G = \{0,1,...,N\}^2,$ we define a tile as a subset of the grid containing a $6 \times 6$ sub-square (containing consecutive points of the grid). Each one of the 36 points of the tile will have a color (black, green or blue) and a direction (up, down, left or right).
    We will identify each tile with the ordered pair $\<u,v\>$ for some $u,v \in \{0,1,...,N-1\}.$
\end{definition}

Our path gadgets will consist of some particular tiles, which are inspired by the work of \cite{fearnley2021complexity}. We will use in total 9 different tiles.

First, let us define the tiles that correspond to the background, i.e., points that are not related with the function that we define.
\begin{definition}
    [Background]
    Given a grid $G = \{0,1,...,N\}^2,$ the background consists of the tiles appearing in \Cref{fig:background}. 
    \begin{enumerate}
        \item In the case that the tile $\<u,v\>$ does not appear in the left layer of the grid, i.e., $u \neq 0$, we will always use the tile of \Cref{fig:back1}.
        \item In the case that the tile $\<u,v\>$ appears in the left layer of the grid, i.e., $u = 0$, we will  use the tile of \Cref{fig:back2}.
    \end{enumerate}
\end{definition}
\begin{figure}[!ht]
    \centering
   \begin{subfigure}[b]{0.45\linewidth}
        \centering
        \includegraphics[width=\linewidth]{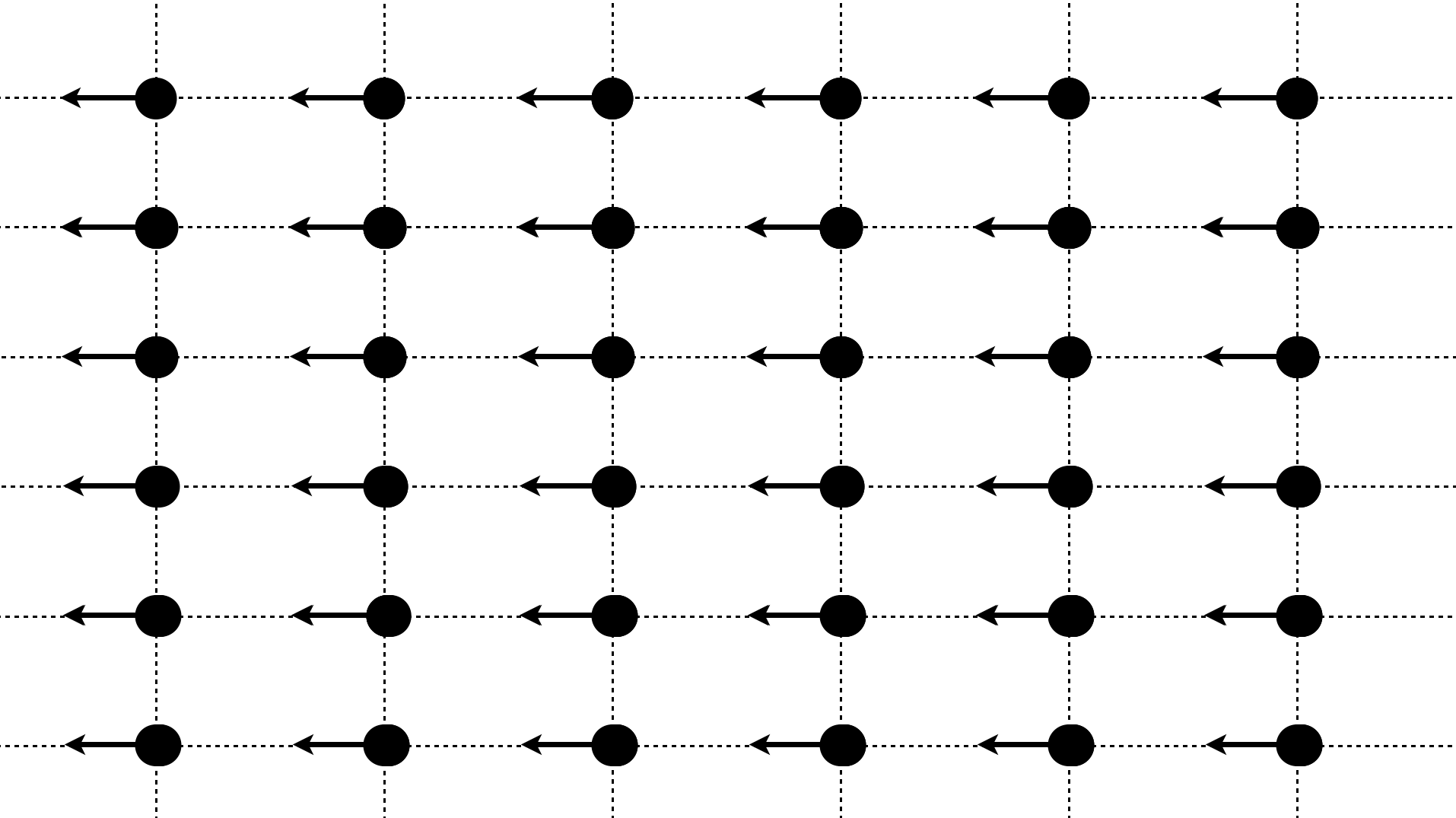}
        \caption{Background going left}
        \label{fig:back1}
    \end{subfigure}
    \hfill
    \begin{subfigure}[b]{0.45\linewidth}
        \centering
        \includegraphics[width=\linewidth]{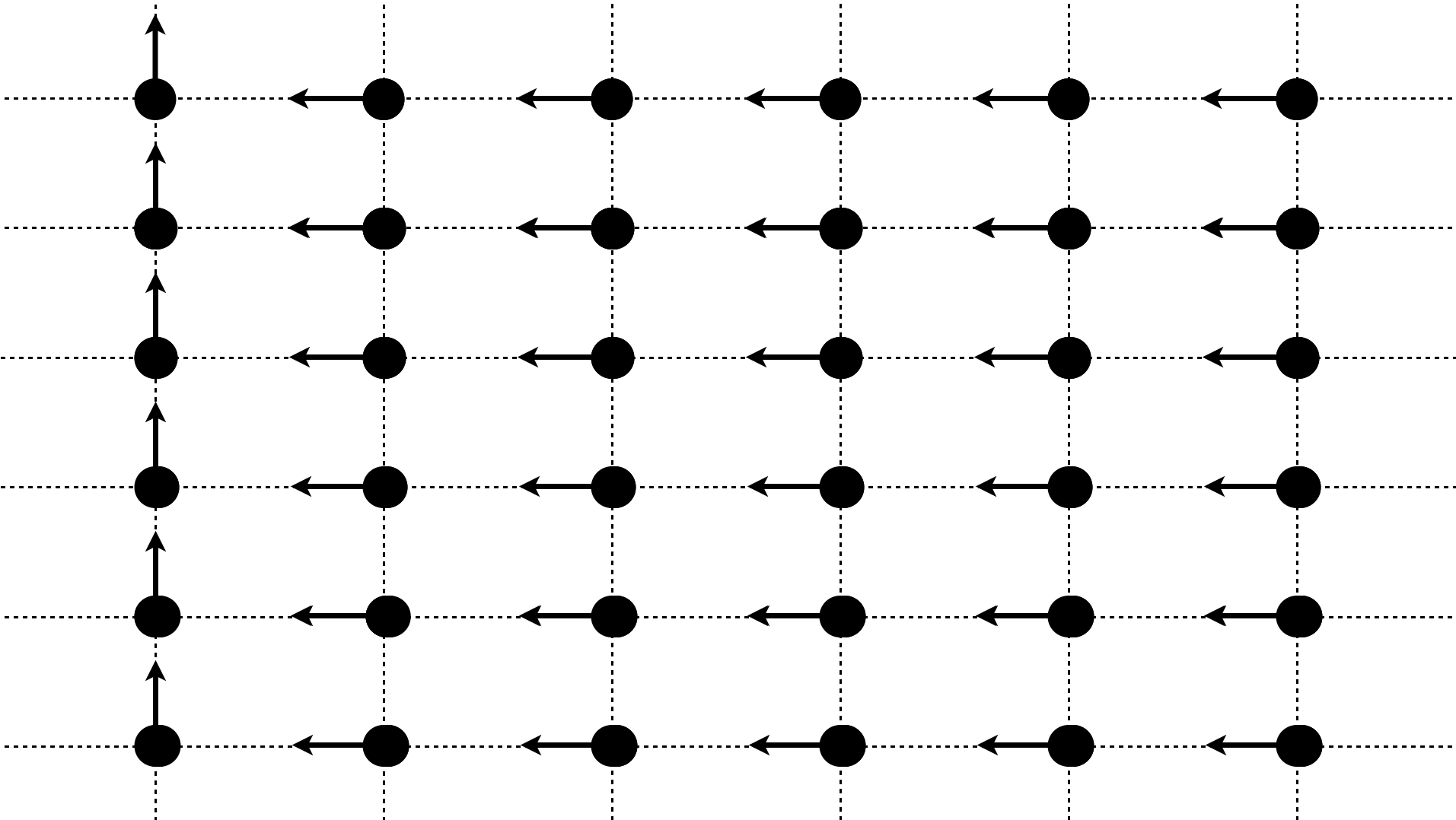}
        \caption{Background meeting the left boundary}
        \label{fig:back2}
    \end{subfigure}
    
    \caption{Background Tiles}
    \label{fig:background}
\end{figure}

Next, we proceed with the definition of monotone paths. Recall that after the embedding of EOPL, there are two kinds of monotone paths. The path with source the given source $v^\star$ of the EOPL instance. This path is special since it start from $\<0,0\>$ and its source should not be a solution (the sink is the only solution). The second kind consists of path that start  from the ``interior'' of the grid, and whose source and sink are solutions.
\begin{definition}
    [Monotone Path]
    Given a grid $G = \{0,1,...,N\}^2,$ we  define a monotone path on it as a list of neighboring tiles with either
    \begin{enumerate}
        \item source placed at the positions $\<0,0\>$ and $\<0,1\>$ of $G$ using the tiles of  \Cref{fig:source},
        \item  then moving straight either right (as in \Cref{fig:right}) or up (as in \Cref{fig:up}) or  
        making a right-up turn (as in \Cref{fig:RU}) or making an up-right turn (as in \Cref{fig:UR}),
        \item until reaching the sink tile of \Cref{figure:sink} ({The sink must be entered from below, because its gadget has only a bottom entrance}).
    \end{enumerate}
    or 
    \begin{enumerate}
        \item source placed at some position $\<i,i\>$ of $G$  for $i > 0$ using the tile of  \Cref{fig:right}.
        \item  then moving straight either right (as in \Cref{fig:right}) or up (as in \Cref{fig:up}) or  
        making a right-up turn (as in \Cref{fig:RU}) or making an up-right turn (as in \Cref{fig:UR}),
        \item until reaching the sink tile of \Cref{figure:sink} ({The sink must be entered from below, because its gadget has only a bottom entrance}).
    \end{enumerate}
\end{definition}

An example of a monotone path can be found in \Cref{fig:example}. For this example, the source is not a solution. On the other side, \Cref{fig:case2} gives a monotone path that is ''inside'' the grid and where both source and sink are solutions.

Let us also comment on how to treat the input source $v^\star.$ Assume that EOPL has given source $v^\star$ and let $b(v^\star)$ be its sub-square in the diagonal of the grid. For our construction, we connect Tile $\<0,0\>$ and Tile $\<0,1\>$ with the $b(v^\star)$ sub-square in order to avoid creating solutions. This is done by first extending Tile $\<0,1\>$ to the right until reaching the column of $b(v^\star)$ and then extending the path vertically until reaching the given source of EOPL. We add the following remark for more details.

\begin{remark}[Embedding the distinguished source]
\label{rem:distinguished-source}
Let $v^\star$ be the distinguished source of the EOPL instance with
diagonal block in the 2D grid $\< \rho(v^\star),\rho(v^\star) \>$. Since $v^\star$ is not a valid solution of EOPL, we connect a
special boundary source to this block, so that $v^\star$ becomes an
internal point of the embedded paths. More precisely, let $c^\star=3\rho(v^\star)+1$, so that $\< c^\star,c^\star\>$ is the central tile of the 
block representing $v^\star$. We place the special source gadgets of
\Cref{fig:source} in tiles $\langle0,0\rangle$ and $\langle0,1\rangle$. Then, we connect them as follows:
\[
\langle0,0\rangle
\longrightarrow \langle0,1\rangle
\xrightarrow{\text{right}} \langle c^\star,1\rangle
\xrightarrow{\text{up}} \langle c^\star,c^\star\rangle
\xrightarrow{\text{right}} \cdots .
\]
All pairs in this display correspond to tile indices. Any crossings of the added path with other paths are resolved using
the local rerouting construction of \Cref{sec:pathCrossings}.
\end{remark}

\begin{figure}[!ht]
    \centering
   \begin{subfigure}[b]{0.45\linewidth}
        \centering
        \includegraphics[width=\linewidth]{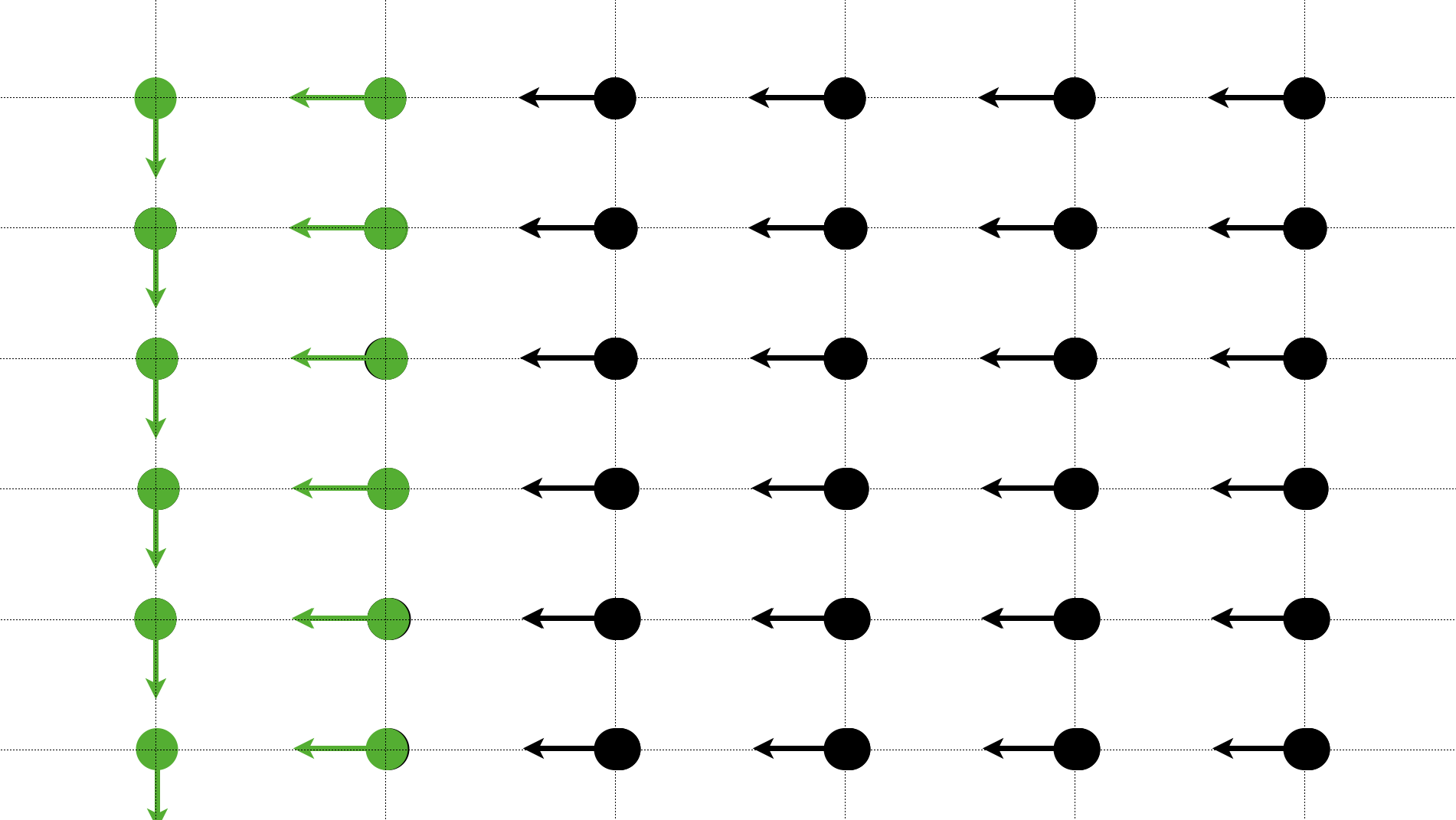}
        \caption{Tile $\<0,0\>$}
        \label{fig:source1}
    \end{subfigure}
    \hfill
    \begin{subfigure}[b]{0.45\linewidth}
        \centering
        \includegraphics[width=\linewidth]{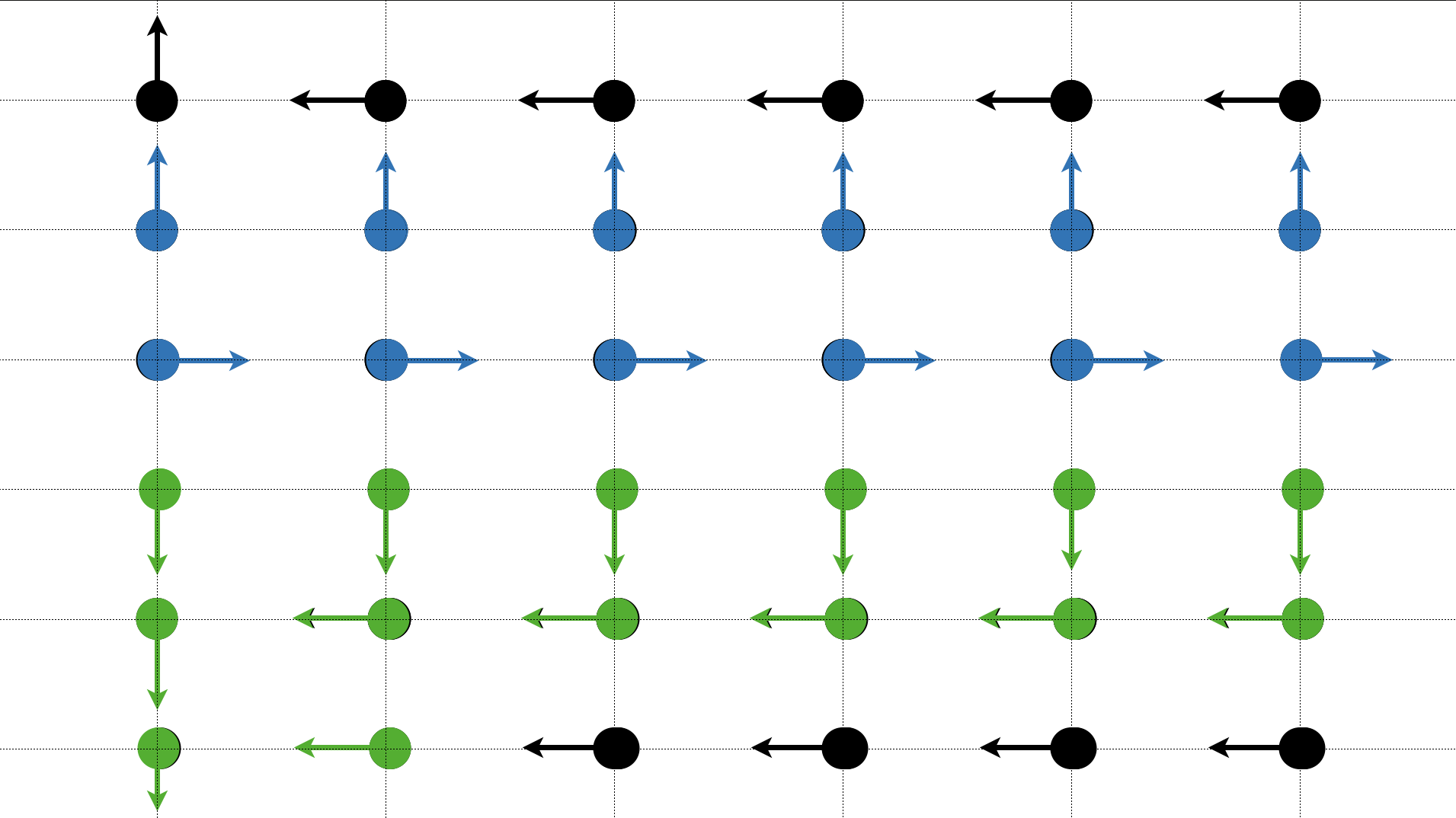}
        \caption{Tile $\<0,1\>$}
        \label{fig:source2}
    \end{subfigure}
    
    \caption{Source of Monotone Path}
    \label{fig:source}
\end{figure}

\begin{figure}[!ht]
    \centering
   \begin{subfigure}[b]{0.45\linewidth}
        \centering
        \includegraphics[width=\linewidth]{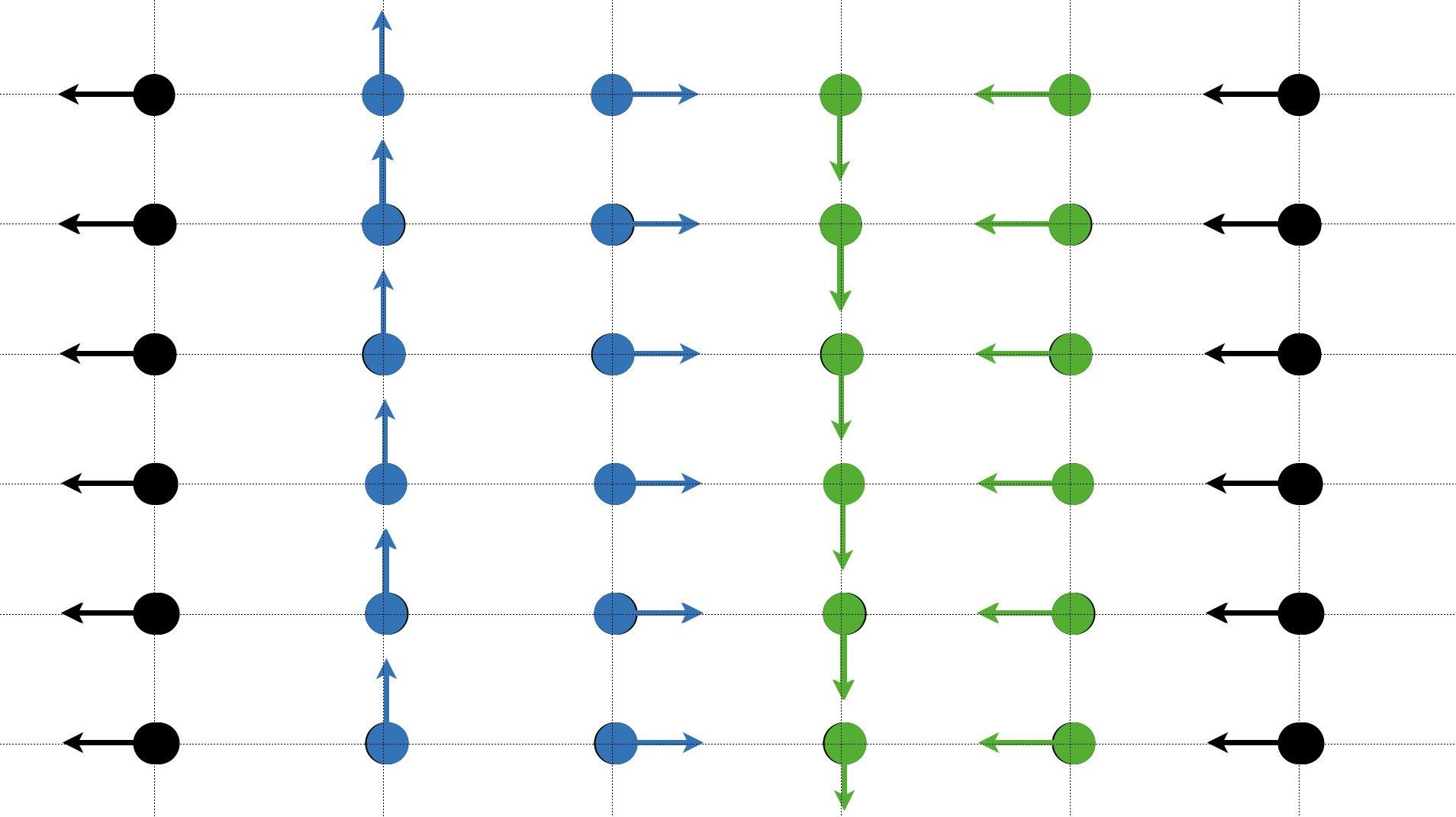}
        \caption{Path I: Up}
        \label{fig:up}
    \end{subfigure}
    \hfill
    \begin{subfigure}[b]{0.45\linewidth}
        \centering
        \includegraphics[width=\linewidth]{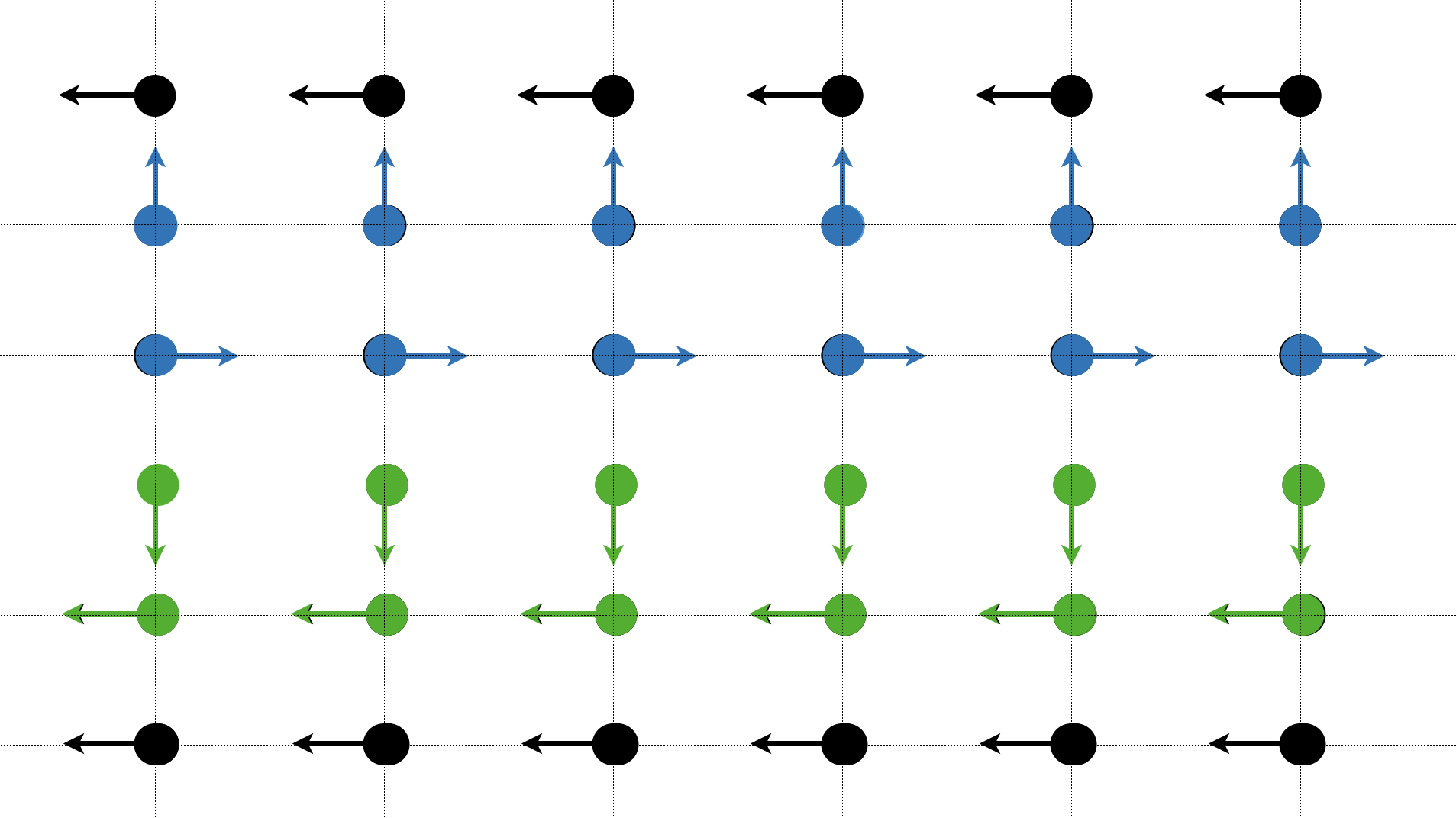}
        \caption{Path II: Right}
        \label{fig:right}
    \end{subfigure}
    
    \caption{Lines of Monotone Path}
    \label{fig:lines}
\end{figure}

\begin{figure}[!ht]
    \centering
   \begin{subfigure}[b]{0.45\linewidth}
        \centering
        \includegraphics[width=\linewidth]{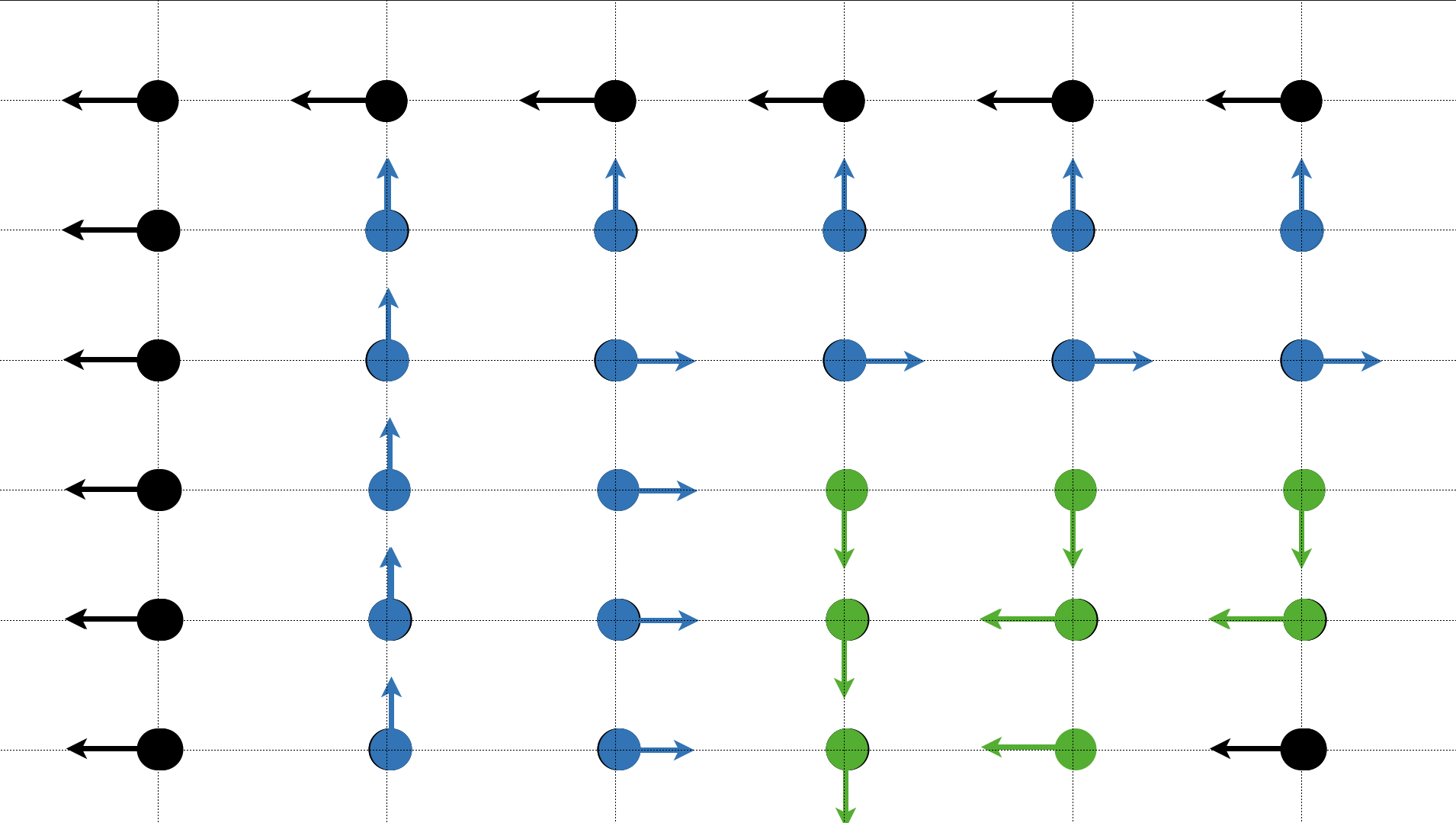}
        \caption{Turn I: Up and Right}
        \label{fig:UR}
    \end{subfigure}
    \hfill
    \begin{subfigure}[b]{0.45\linewidth}
        \centering
        \includegraphics[width=\linewidth]{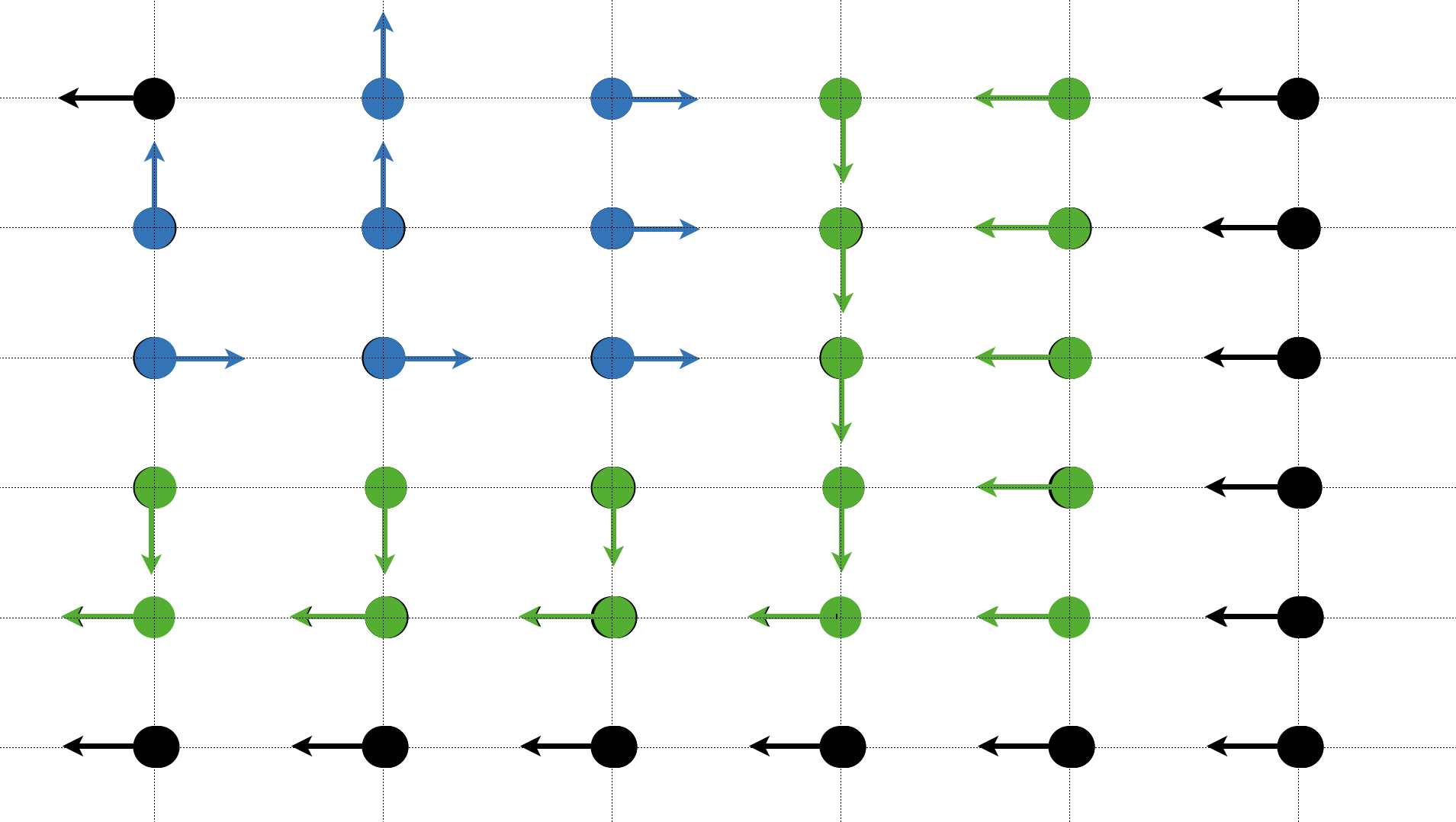}
        \caption{Turn II: Right and Up}
        \label{fig:RU}
    \end{subfigure}
    
    \caption{Turns of Monotone Path}
    \label{fig:turns}
\end{figure}

\begin{figure}[!ht]
    \centering
    \includegraphics[scale=0.2]{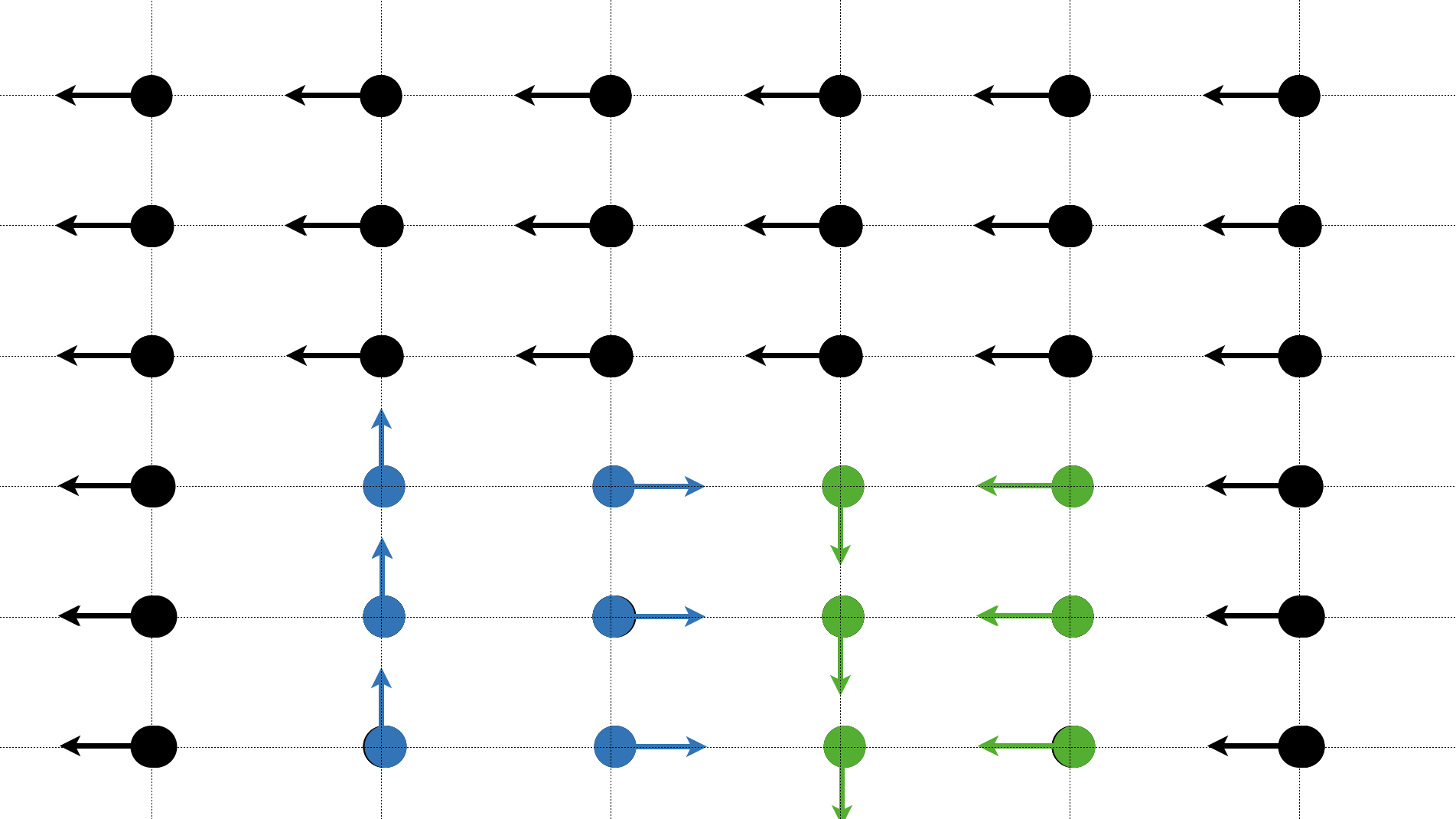}
    \caption{Sink of Monotone Path}
    \label{figure:sink}
\end{figure}

\begin{figure}[!ht]
    \centering
    \includegraphics[width=0.7\linewidth]{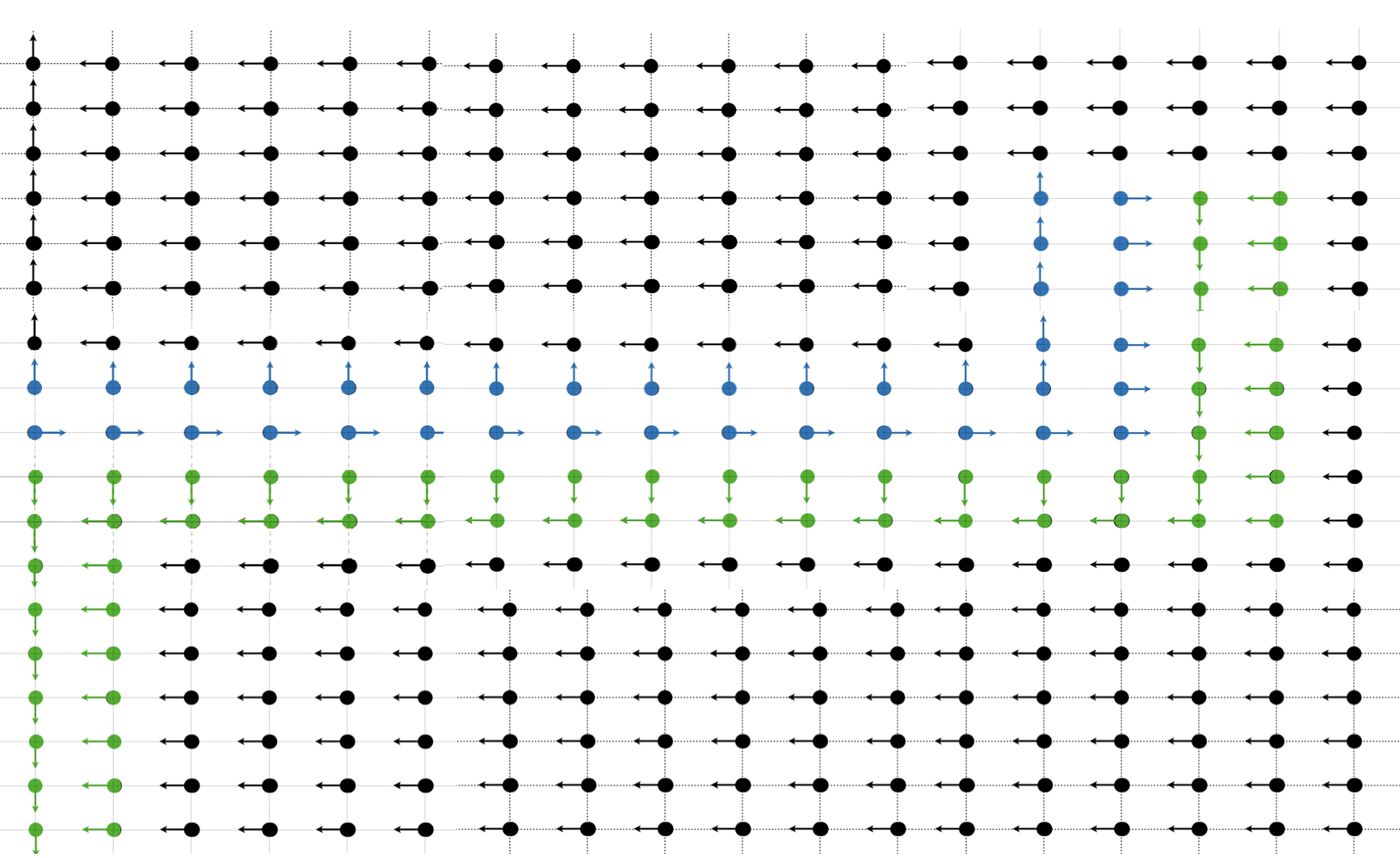}
    \caption{An Example of a Monotone Path on a 18 $\times $ 18 grid using 9 tiles.
    The source of this path is the input source of EOPL and hence does not contain a GDA fixed point (only the sink contains a solution).
    In $\<0,0\>$ and $\<0,1\>$, we use the source tiles, in $\<0,2\>$ we use the background tile that touches the left layer, and in positions $\<1,0\>, \<2,0\>, \<1,2\>$, we use the background tile that is going left. Finally, in position $\<1,1\>$ we use the tile where the path is moving right, in $\<1,2\>$ the path is going right and up and finally in $\<2,2\>$ we place the sink tile.}
    \label{fig:example}
\end{figure}

\begin{figure}[!ht]
    \centering
    \includegraphics[width=0.7\linewidth]{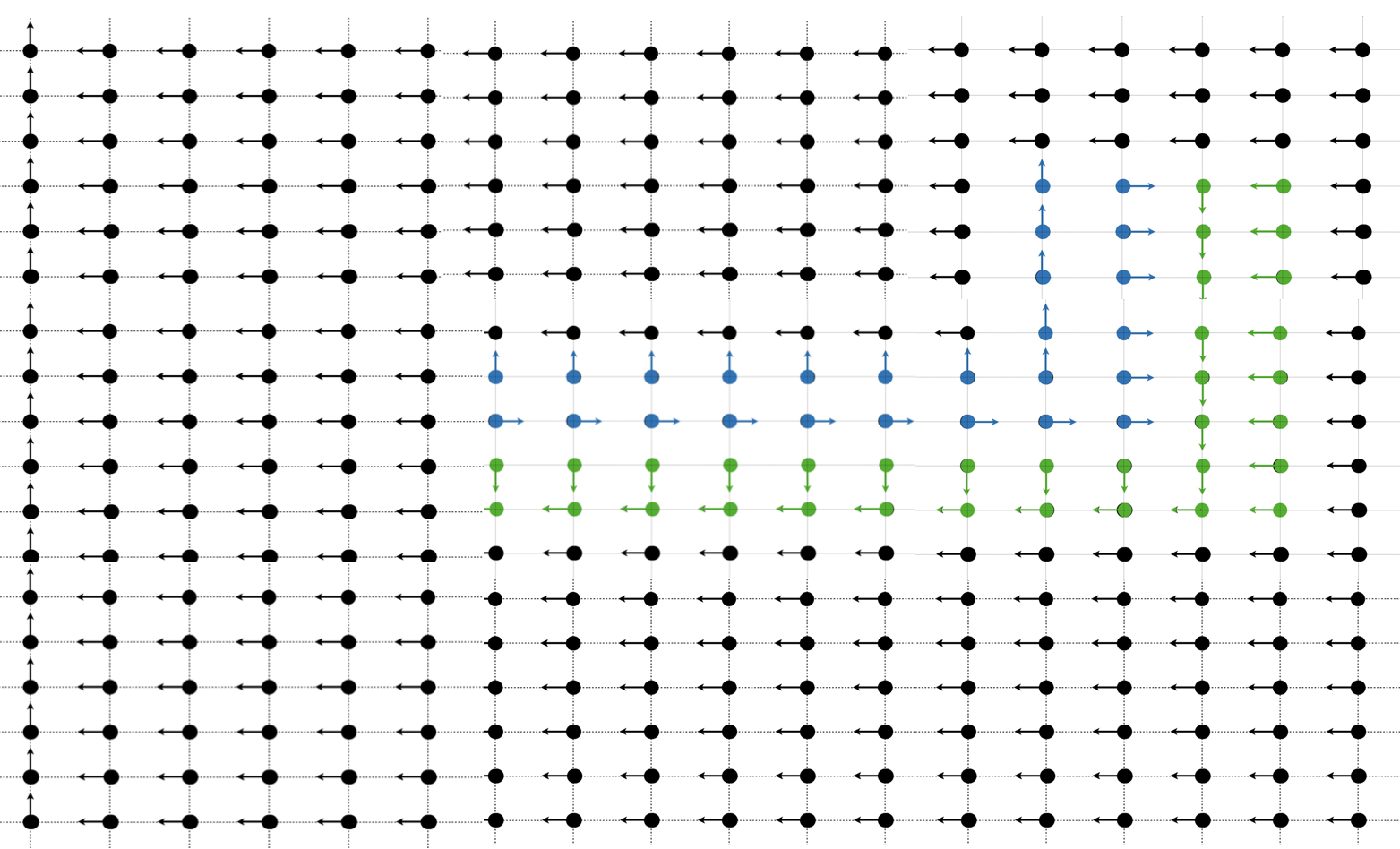}
    \caption{An Example of a Monotone Path on a 18 $\times $ 18 grid using 9 tiles. In this example, both the source and the sink sub-squares contain a solution and this is also the case in EOPL, since it corresponds to a non-trivial source.}
    \label{fig:case2}
\end{figure}

Given the above collection of tiles, we can define monotone paths on the grid together with the background tiles. Any such path is associated with a function $f$ defined over the grid $G$. The values of the function at a point $(x,y) \in G$ are determined by the colors of the point (black, green and blue) based on \Cref{sec:function-values} and the gradient value is determined by the direction assigned (as shown in the tiles).

\begin{remark}
[Connection to \cite{fearnley2021complexity}]
In both our work and \cite{fearnley2021complexity}, the arrows in the construction correspond to the direction of $-\nabla f$, i.e., the direction where the function value decreases. However, the solution concepts are different. In \cite{fearnley2021complexity}, we can think that there are two minimization players that want to 
decrease the function value and hence move according to the arrows both horizontally and vertically. Hence, for instance, in \Cref{fig:example}, the bottom left square has a KKT solution. On the other side, we are looking for fixed points of GDA. In our case, the min player, that is moving horizontally, wants to follow the arrow since it decreases the function value. 
However, the max player is moving vertically and wants to increase the function value; so she wants to move in the opposite direction of the arrow. 
\end{remark}

\subsubsection{Handling Path Crossings}
\label{sec:pathCrossings}
By our construction, there might exist points
where the monotone paths between different connected components of the EOPL graph cross. First of all, note that it is impossible to have more than two paths
traversing a sub-square, and thus any crossing involves exactly two paths. Furthermore, no crossing can occur in big
squares where a “turn” occurs, since, in that case, the turn connects the two paths.
The only way for two paths to cross is the case where a  path traverses a big square from left to right,
and a second path traverses the same big square from bottom to top. This problem always occurs when one tries to embed an End-of-Line instance in a two-dimensional domain and can be solved using re-routing the paths \cite{chen2009complexity,fearnley2021complexity}.
The idea of Chen and Deng \cite{chen2009complexity} is to locally re-route the two
paths so that they no longer cross. This modification has the following two crucial properties: a) it is completely local,
and b) it does not introduce any new solution (in our case a fixed point of GDA). This is because, even if the resulting paths change after the local re-routing, the solutions in the routed instance are exactly the same as in the original instance. We refer to Figure 10(a) of \cite{fearnley2021complexity} for details on how to do the re-routing.

{
In our construction, we implement this rerouting by the following
$3\times3$ block, whose entries denote the tile panels in Figures~\ref{fig:background}--\ref{fig:turns}
and whose rows are listed from top to bottom:
\[
\begin{array}{ccc}
  \ref{fig:back1} & \ref{fig:up} & \ref{fig:back1}\\
  \ref{fig:right} & \ref{fig:RU} & \ref{fig:UR}\\
  \ref{fig:back1} & \ref{fig:UR} & \ref{fig:RU}
\end{array}
\]
It connects the left entrance to the top exit and the bottom entrance
to the right exit, with matching tile boundaries and no new endpoint.
The required constant-factor refinement of the embedding is absorbed
in~$N$.
}

Finally, any sub-square that is not traversed by any path (including the isolated vertices of the End-of-Potential-Line instance) will take the background coloring.

\subsection{Extending the function to the continuous grid
}
\label{sec:extension}

For now, we have only defined the function $f$ and the direction of its gradient at
all grid points of $G$. In order to extend $f$ to the whole domain $[0, N]^2$, we use bicubic interpolation. Consider any sub-square and a point $(x,y) \in [0,1]^2$ inside it. Then, the bicubic interpolation inside this sub-square will be a polynomial of the form:
\begin{equation}\label{eq:bicubic}
f(x,y) = \sum_{i=0}^3 \sum_{j=0}^3 a_{ij} x^i y^j
\end{equation}
where the coefficients $a_{ij}$ depend on the function values and the gradients in the four corners of the sub-square ($00,01,10$, and $11$), and are computed as follows \footnote{See also \url{https://en.wikipedia.org/wiki/Bicubic_interpolation} and \cite{fearnley2021complexity}.}
\begin{align}\label{eq:bicubic-coefficients}
&\begin{bmatrix}
a_{00} & a_{01} & a_{02} & a_{03}\\
a_{10} & a_{11} & a_{12} & a_{13}\\
a_{20} & a_{21} & a_{22} & a_{23}\\
a_{30} & a_{31} & a_{32} & a_{33}\\
\end{bmatrix}\\
&=
\begin{bmatrix}
1 & 0 & 0 & 0\\
0 & 0 & 1 & 0\\
-3 & 3 & -2 & -1\\
2 & -2 & 1 & 1\\
\end{bmatrix} \cdot 
\begin{bmatrix}
f(0,0) & f(0,1) & f_y(0,0) & f_y(0,1)\\
f(1,0) & f(1,1) & f_y(1,0) & f_y(1,1)\\
f_x(0,0) & f_x(0,1) & f_{xy}(0,0) & f_{xy}(0,1)\\
f_x(1,0) & f_x(1,1) & f_{xy}(1,0) & f_{xy}(1,1)\\
\end{bmatrix}
\cdot
\begin{bmatrix}
1 & 0 & -3 & 2\\
0 & 0 & 3 & -2\\
0 & 1 & -2 & 1\\
0 & 0 & -1 & 1\\
\end{bmatrix}\,.
\end{align}
In the above, $f_x$ (resp. $f_y$) denote the partial derivative of $f$ with respect to $x$ (resp. $y$) and $f_{xy}$ denotes the second order partial derivative of $f$ with respect to $x$ and $y$. 
We set the values of $f, f_x, f_y$ and $f_{xy}$ at the four corners of the square in a similar way as \cite{fearnley2021complexity}:
\begin{itemize}
	\item The values $f(0,0), f(0,1), f(1,0)$ and $f(1,1)$ are set according to the value regimes in our construction, i.e., the colors of the corner points. For example if the point $(0,0)$ (wrt to the sub-square) has coordinates $(x,y)$ in the grid and is colored black, then $f(0,0) = x - y + 4N + 10$. 
	
    \item The values of $f_x(0,0), f_x(0,1), f_x(1,0), f_x(1,1), f_y(0,0), f_y(0,1), f_y(1,0)$ and $f_y(1,1)$ are set based on the direction of steepest descent ($- \nabla f$) in our construction, with a length multiplier of $\delta = 1/2$.
    It is important to note that the arrows in our construction indicate the direction of $-\nabla f$ (function decrease). Hence, the min player, that is moving horizontally, wants to follow the arrow since it decreases the function value. 
    On the other side, the max player, that is moving vertically, wants to increase the function value and hence wants to move in the opposite direction of the arrow.
    For example, if the arrow of steepest descent at $(0,1)$ is pointing to the left, then we set $f_x(0,1)=\delta$ and $f_y(0,1)=0$. If it is pointing up, then we set $f_x(0,1)=0$ and $f_y(0,1)=-\delta$.

	\item We always set $f_{xy}(0,0) = f_{xy}(0,1) = f_{xy}(1,0) = f_{xy}(1,1) = 0$.
\end{itemize}

This yields a function $f : [0,N]^2 \to \reals$ (or more generally $\reals^2 \to \reals)$ which is continuously differentiable in $\reals^2$. Also, since the extension will be a polynomial, both $f$ and its gradient $\nabla f$ will be Lipschitz-continuous on $[0, N]^2$ with Lipschitz constant
$L = \Omega(N)$. This means that $f$ is $\Theta(N)$-smooth. As mentioned in \cite{fearnley2021complexity}, we can extend the function not only on $[0,N]^2$ but also in $\mathbb R^2$ (for the details we refer to their work).

\begin{proposition}
[Extension of the function]
Given a monotone path, the function $f : \reals^2 \to \reals$, obtained by bicubic interpolation on each one of the tiles of $\{0,1,...,N\}^2$, has the following
properties:
\begin{enumerate}
    \item $f$ is continuously differentiable on $\reals^2$, and,
    \item $f$ and its gradient $\nabla f$ are Lipschitz-continuous on $[0, N]^2$ with Lipschitz constant $L = \Theta(N)$.
\end{enumerate}
\end{proposition}
\begin{proof}
Let $f_x, f_y$ denote the partial derivatives with respect to $x$ and $y$ respectively and $f_{xy}$ 
denotes the second
order partial derivative with respect to $x$ and $y$. The proof is essentially the same as in \cite[Lemma 4.2]{fearnley2021complexity} but in our case the values of $f, f_x, f_y, f_{xy}$ are bounded by $2^4 N$ instead of $2^3 N$.
\end{proof}

\subsection{Correctness}
\label{sec:valid}

{We have worked on the domain $[0,N]^2$. For this section, we will work with constant approximation accuracy $\epsilon = 0.01$.}
In order for the above construction to be correct we have to show that 0.01-first order Nash points of $f$ only lie:
\begin{enumerate}
    \item either at a  source sub-square of some monotone path with source not corresponding to the one given as input to EOPL
    \item  or at some sink sub-square (the sub-square at the end of the monotone path).
\end{enumerate}
As we discussed before, in these sub-squares we embedded the solutions of EOPL. Now we have to show that our construction did not create any other approximate fixed points of GDA.

{
We divide the proof into two parts. The reason is that an $\epsilon$-approximate fixed point of $\mathsf{GDA}$ of $f$ is equivalent to an $\epsilon$-$\mathsf{KKT}$ point of $f$ in the interior of the domain. This is because, as long as
$(x  - \nabla_x f(x, y), y + \nabla_y f(x, y)) \in [0,N] \times [0,N]$, then the fixed point condition reads as
\begin{equation}\label{eq:int-KKT}
\|\nabla_x f(x,y) \| \leq \epsilon\,,
\|\nabla_y f(x,y) \| \leq \epsilon\,.
\end{equation}
Hence, $(x,y)$ is an $\epsilon$-$\mathsf{KKT}$ of $f$. The work of \cite{fearnley2021complexity} provides sufficient conditions on the values/coloring of each grid point such as the function the function $f$ constructed by the bicubic interpolation does not admit any interior $\epsilon$-KKT point~(Equation~\ref{eq:int-KKT}). Our construction abides with these conditions apart from sink squares or non-trivial source squares. Hence, it remains to show that the desired property only for the points for which the projection operation is needed either in $X$ or $Y$. 

\begin{definition}
[Good Points]
For a function $f : [0,N]^2 \to \mathbb R$, we will say that $(x,y) \in [0,N]^2$ is good if
\[
(x  - \nabla_x f(x, y), y + \nabla_y f(x, y)) \in [0,N] \times [0,N]\,.
\]
\end{definition}

\begin{fact}
    For a function $f : [0,N]^2 \to \mathbb R$, let $(x,y) \in [0,N]^2$ be a good point. Then 
    if it is an $\epsilon$-approximate fixed point of $\mathsf{GDA}$ of $f$, then it 
    is also an $\epsilon$-$\mathsf{KKT}$ of $f$.
\end{fact}

A last observation is in order. Let us assume that $(x,y)$ is not good and direction $x$ is the one that violates the goodness of $(x,y)$. Since we are working on $[0,N]^2$, this means that either 
$x - \nabla_x f(x,y) < 0$ and so 
$\Pi_{[0,N]}(x - \nabla_x f(x,y)) = 0$ or
$x - \nabla_x f(x,y) > N$ and so 
$\Pi_{[0,N]}(x - \nabla_x f(x,y))  = N$. Let us assume that $(x,y)$ is close to the left boundary. Then $(x,y)$ is not an $\epsilon$-fixed point of $\mathsf{GDA}$ of $f$ is either $x > \epsilon$ or $\|\nabla_y f(x,y)\| > \epsilon$.  

In summary, we have the following connection between fixed points of gradient descent/ascent of $f$ and the goodness of a  point $(x,y)$.

\begin{observation}
    Let $(x,y)$ be an $\epsilon$-approximate fixed point of $\mathsf{GDA}$ of $f$ then one of the following holds:
\begin{enumerate}
    \item If $(x,y)$ is good, then  $\|\nabla_x f(x,y)\|, \|\nabla_y f(x,y) \| \leq \epsilon$

    \item If $(x,y)$ is bad, then either
    \begin{itemize}
    \item $\|\nabla_x f(x,y)\| < \epsilon$ and $y \in [0,\epsilon] \cup [N-\epsilon, N]$

    \item $\|\nabla_y f(x,y)\| < \epsilon$ and $x \in [0,\epsilon] \cup [N-\epsilon, N],$

    \item $x,y \in [0,\epsilon] \cup [N-\epsilon, N].$
    
    \end{itemize}
\end{enumerate}
\label{obs:good-points}
\end{observation}

This gives us a way to look for approximate fixed points of GDA for $f$. We take two cases based on whether a point is good or bad.

\paragraph{Roadmap} We decompose the proof into two cases.
\begin{enumerate}
        \item In the first case (no active projections), the point will be a good point of $f$. This means that is suffices to show that it is not a KKT point of $f$. In \Cref{sec:proof1}, we prove that this is indeed the case. 

        \item In the second case, one (or both) of the projections gets activated, and hence the point is bad for $f$. For bad point to be solutions, we observe that they should lie to the outer region of the grid. This means that is suffices to show that in any such sub-square, our constructed $f$ does not have approximate GDA fixed points.
    \end{enumerate}

\begin{definition}
    Given an EOPL instance with input source $v^\star$, we call a source non-trivial if it is a source of the EOPL instance different than $v^\star.$
\end{definition}

For the localization statements below, the solution sub-square of a
non-trivial source includes the elementary squares straddling the left
entrance of its initial right-path tile. Such a square is decoded using
the source label of that adjacent tile.

We will prove the following claims.
\begin{claim}
    Let $\epsilon = 0.01$.
    Let $(x,y)$ be an $\epsilon$-approximate fixed point of GDA for $f$. If $(x,y)$ is good, then it should be at a sink or a non-trivial source sub-square.
\end{claim}

\begin{claim}
     Let $\epsilon = 0.01$.
     Let $(x,y)$ be an $\epsilon$-approximate fixed point of GDA for $f$.
     If $(x,y)$ is bad, then it should lie at the $\epsilon$-outer layer of the grid and there are no  $\epsilon$-approximate fixed points of GDA for $f$ in this region.
\end{claim}

}

\subsubsection{Analysis of Good Points}
\label{sec:proof1}
In this section, we prove the following result.

\begin{lemma}
\label{lemma:sinks}
    Let $\epsilon = 0.01$.
    Given an EOPL instance with input source $v^\star$, let $f: [0,N]^2 \to \mathbb R$ be the associated function of \Cref{sec:extension}.
    Let $(x,y)$ be a good point of $f$ that is not in any sink or any non-trivial source sub-square. Then 
   $(x,y)$ is not an $\epsilon$-approximate fixed point of $\mathsf{GDA}$ of $f$.
\end{lemma}

\begin{proof}
Given an EOPL instance with input source $v^\star$, we call a sub-square a solution sub-square if this corresponds to either a sink or a non-trivial source.
We define a \emph{small} square as the square induced by adjacent points of the grid of a sub-square.
We will show that any small square  that does not lie in a
solution sub-square, does not contain any $0.01$-first order Nash point.
Our proof follows the ideas of \cite[Lemma 4.3]{fearnley2021complexity}.
\cite{fearnley2021complexity}
considered any small square with the color and arrow information about its four
corners, without any information about its position. If we manage to show that
a small square does not contain any $0.01$-first-order Nash point of $f$ using only this information, then this will
always hold, wherever the small square is positioned. As a result, \cite{fearnley2021complexity} obtain a finite number of
squares (with color and arrow information) that need to be checked. We will follow the same strategy.
We will use the groups $G_1, G_2, G_3$  $G_4$, as defined by \cite{fearnley2021complexity} (see \Cref{fig:g1,fig:g2,fig:g3,fig:g4}).

\begin{proposition}
[Lemma 4.3 in \cite{fearnley2021complexity}]

{
Let $\epsilon=0.01$. For any element $g$ of $G_i$, with
$i\in\{1,2,3,4\}$, its bicubic interpolant satisfies
\[
 \max\{|\nabla_x f(x,y)|,|\nabla_y f(x,y)|\}>\epsilon
\]
throughout the closed small square. In particular, the bicubic interpolation on g does not introduce interior $\eps$-$\mathsf{KKT}$ points of $f$.
}
\end{proposition}

{
It remains to verify that every small square in the construction of
\Cref{sec:grids}, outside the sink and non-trivial source solution
sub-squares, satisfies the conditions of one of $G_1,G_2,G_3,G_4$.
Figures~\ref{fig:g1}--\ref{fig:g4} list the relevant local configurations and their
generators.
}

\begin{figure}[!ht]
    \centering
    \includegraphics[scale = 0.3]{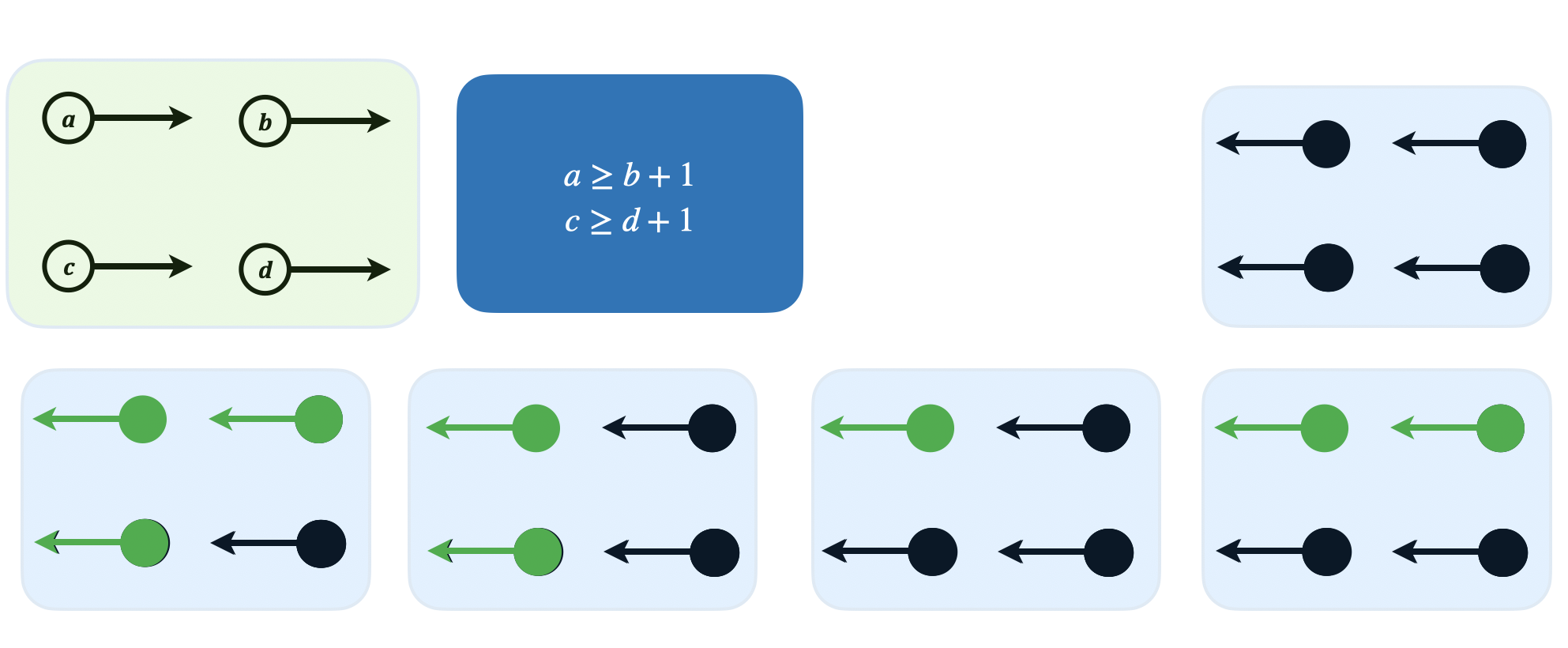}
    \caption{Generator of Group 1 with the corresponding inequality constraints and the associated tilings of our construction in the light blue boxes.}
    \label{fig:g1}
\end{figure}

\begin{figure}[!ht]
    \centering
    \includegraphics[scale = 0.3]{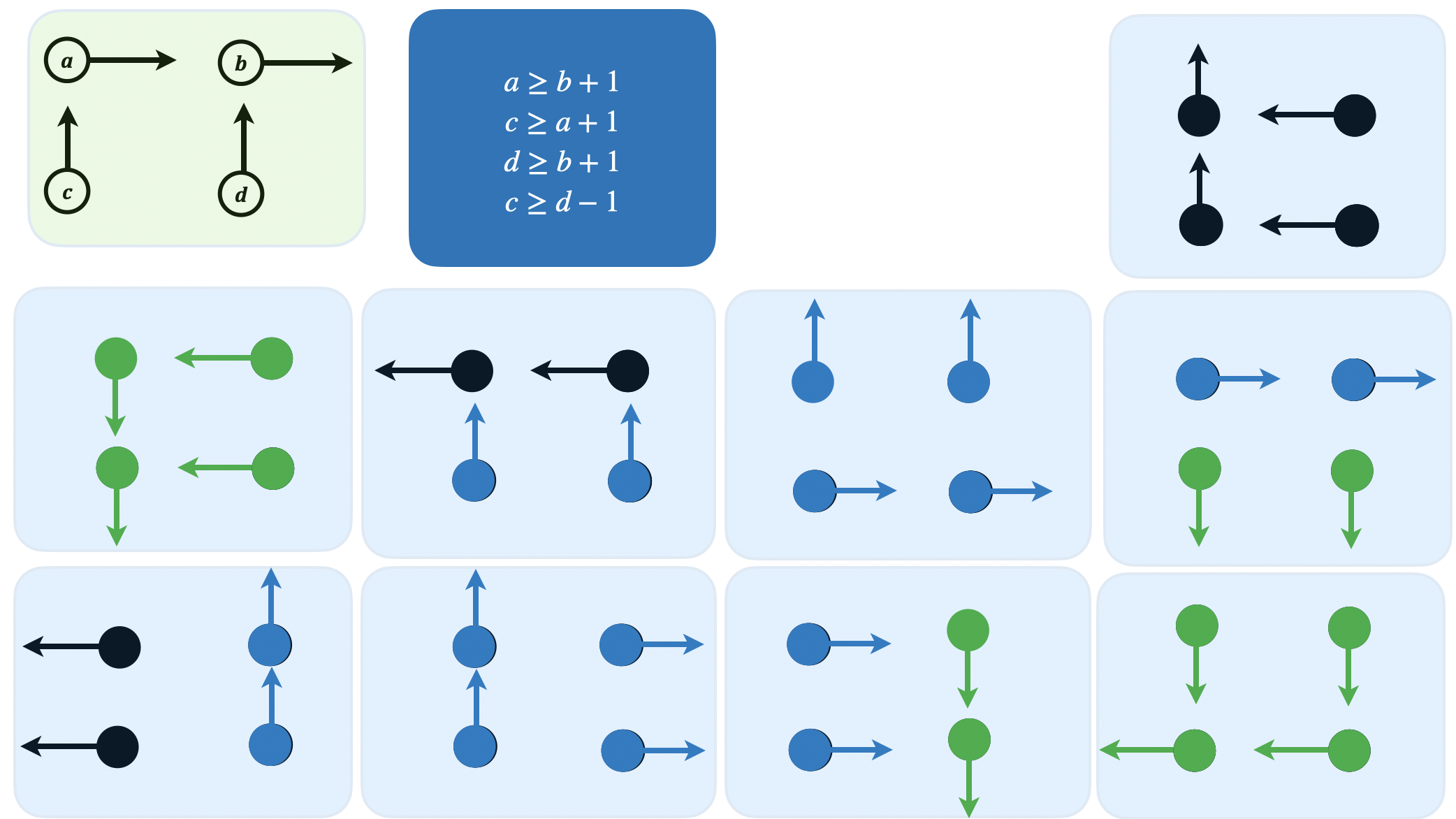}
    \caption{Generator of Group 2 with the corresponding inequality constraints and the associated tilings of our construction in the light blue boxes.}
    \label{fig:g2}
\end{figure}

\begin{figure}[!ht]
    \centering
    \includegraphics[scale = 0.3]{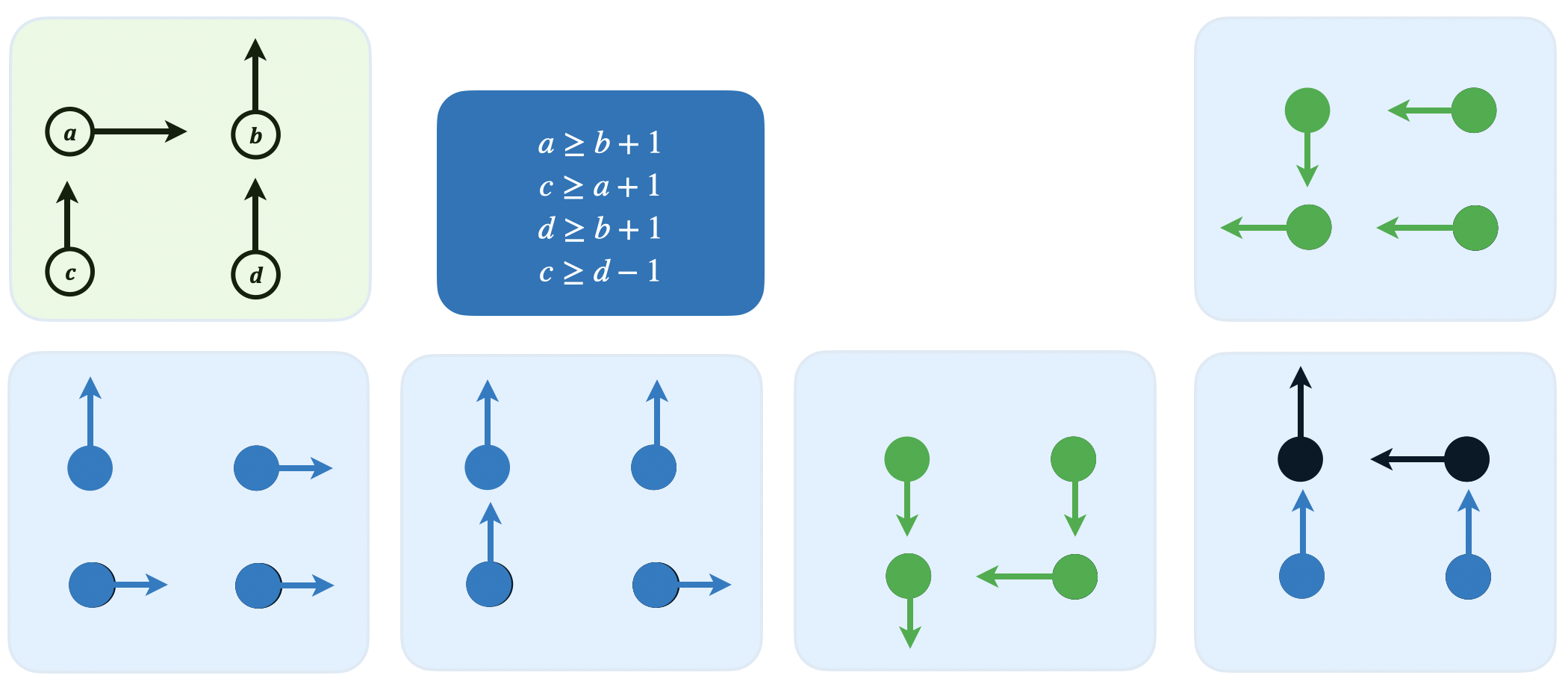}
    \caption{Generator of Group 3 with the corresponding inequality constraints and the associated tilings of our construction in the light blue boxes.}
    \label{fig:g3}
\end{figure}

\begin{figure}[!ht]
    \centering
    \includegraphics[scale=0.3]{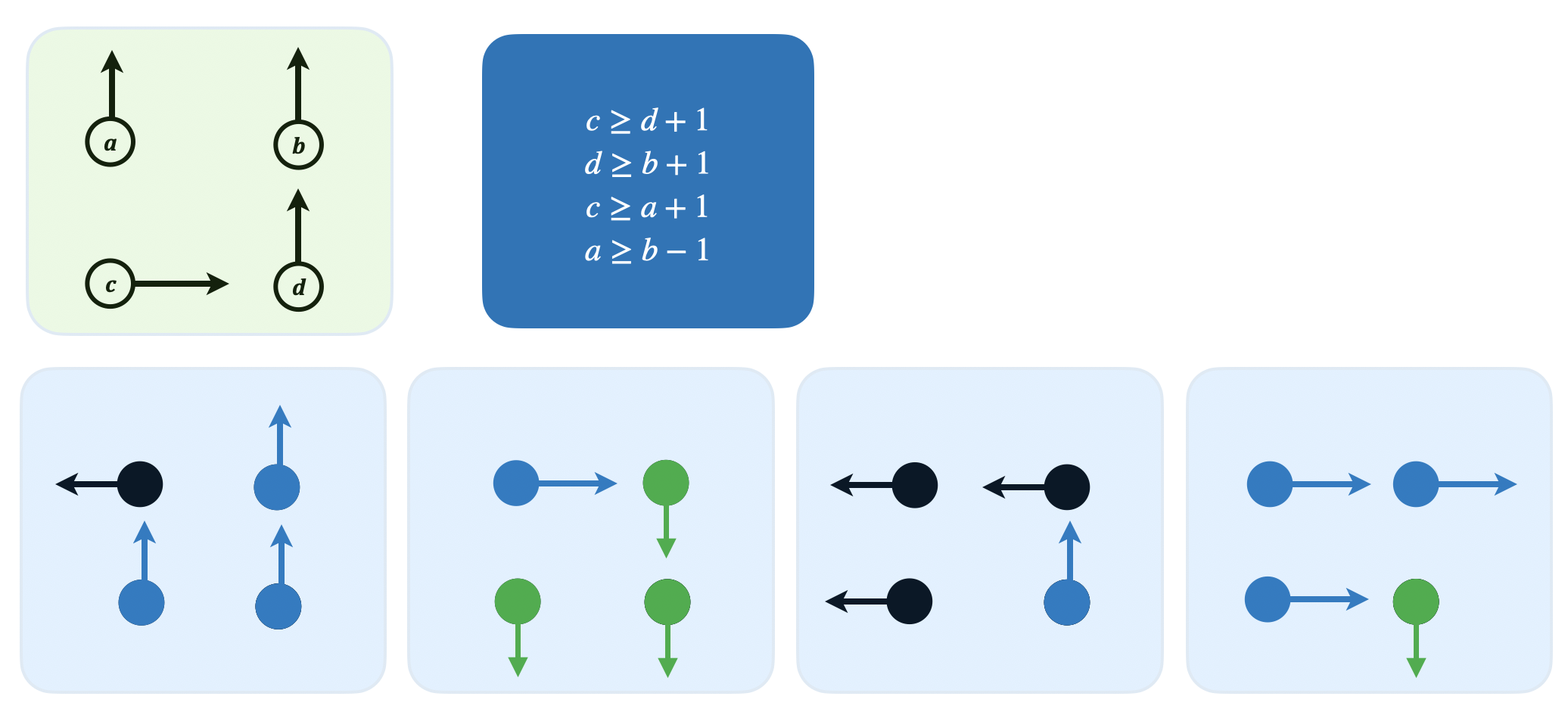}
    \caption{Generator of Group 4 with the corresponding inequality constraints and the associated tilings of our construction in the light blue boxes.
    }
    \label{fig:g4}
\end{figure}

\begin{proposition}

{
Every small square outside the sink and non-trivial source solution
sub-squares belongs to one of the groups $G_i$, $i\in\{1,2,3,4\}$.
}
\end{proposition}
\begin{proof}
    The proof follows by construction and the properties of the groups of \cite{fearnley2021complexity}.

\end{proof}

{
The preceding two propositions, together with
\Cref{obs:good-points}, complete the proof of
\Cref{lemma:sinks}.
}

\end{proof}

\subsubsection{Analysis of Non-Good Points}
Up to this point we have shown that there are no unintended approximate fixed points of $\mathsf{GDA}$
in the interior of the domain $[0, N]^2$ of our instance. 
It remains to show that no solutions appear at points that are not good. 

\begin{lemma}
[Location of Bad Points]
 Given an instance of EOPL, let $f: [0,N]^2 \to \mathbb R$ be the associated function of \Cref{sec:extension}.
 Then the bad points of $f$ that are $\epsilon$-approximate fixed points of GDA of $f$ lie in the $\epsilon$-outer layer of the grid.
\end{lemma}

\begin{proof}
Let $(x,y)$ be a bad point and $\epsilon$-approximate point of GDA. We assume that $(x,y)$ does not lie on the outer layer of the $[0,N]^2$ grid and we will show a contradiction. As a result, we will get that such a point must necessarily lie on the outer layer.
\smallskip

\noindent Since $(x,y)$ is a bad point, we know that one of the following $4$ cases holds: $i)$
$\Pi_{[0,N]}(x - \nabla_x f(x,y)) = 0$ or $ii)$ $\Pi_{[0,N]}(x - \nabla_x f(x,y))  = N$ or $\Pi_{[0,N]}(y + \nabla_y f(x,y)) = 0$ or $ii)$ $\Pi_{[0,N]}(y + \nabla_y f(x,y))  = N$. Notice that by the definition of $\epsilon$-approximate GDA fixed points, 
\[
\max \left\{ \| x - \Pi_X(x^\star - \nabla_x f(x, y))\|, 
\| y - \Pi_Y(y + \nabla_y f(x, y))\| \right\}
\leq \epsilon.
\]
However in case $(x,y)$ does not lie on the outer layer of the $[0,N]^2$ grid then 
{
\[
\max \left\{ \| x - \Pi_X(x - \nabla_x f(x, y))\|, 
\| y - \Pi_Y(y + \nabla_y f(x, y))\| \right\} 
> \epsilon.
\]
}
{
That is, since $(x,y)$ is bad, then at least one of  $x  - \nabla_x f(x, y)$ and $y + \nabla_y f(x, y)$ lies outside of $[0,N]$, and so is projected to either $0$ or $N$. Since both $x$ and $y$ lie in $(\epsilon, N-\epsilon)$ the corresponding coordinate changes by more than $\epsilon$, a contradiction.

}

\end{proof}

Next, we will show that by construction of $f$, there are no $\epsilon$-approximate fixed points of GDA of $f$ in the $\epsilon$-outer layer of the grid. 

\begin{observation}
    The outer layer contains 8 different small squares.
\end{observation}

\begin{lemma}
[Bad Points are not fixed points of $\mathsf{GDA}$]
\label{lemma:sinks-boundary}
   Given an EOPL instance, let $f: [0,N]^2 \to \mathbb R$ be the associated function of \Cref{sec:extension}.
    Let $(x,y)$ be a bad point of $f$. Then $(x,y)$ is not an $\epsilon$-approximate fixed point of $\mathsf{GDA}$ of $f$.
\end{lemma}

\noindent The proof of Lemma \ref{lemma:sinks-boundary} is deferred to Appendix \ref{app:cases}.\\
{
\noindent\emph{Proof outline.}
Let $\epsilon=0.01$. If a bad point were an
$\epsilon$-approximate fixed point of $\mathsf{GDA}$, it would
lie within distance $\epsilon$ of a boundary. We examine the eight types of small squares occurring
in this boundary layer. In each case, the bicubic interpolation
formulas identify a coordinate whose proposed GDA displacement
exceeds $\epsilon$, with the domain boundary in that direction
also more than $\epsilon$ away. Thus, even after projection,
that coordinate changes by more than $\epsilon$, leading to a contradiction.
The detailed case analysis is given in \Cref{app:cases}.}

\subsection{Completing the reduction}
\label{subsec:ReductionComplete}

{
We now turn the auxiliary function constructed above into an instance of
$\textsc{2D-GDAFixedPoint}$ on $[0,1]^2$ and complete the reduction.
}

First, it is not hard to show (see \cite{fearnley2021complexity}) that arithmetic circuits for both $f$ and $\nabla f$ can be computed in polynomial time in $n$ and in the sizes of $\mathsf{Succ}, \mathsf{Pred}, \mathsf{Potential}$. In fact, they only use a constant number of true multiplication
gates.{\footnote{Consequently, our $\mathsf{CLS}$-hardness result also holds
when $f$ and $\nabla f$ are represented by well-behaved arithmetic circuits,
as in \cite{fearnley2021complexity}. These circuits admit exact
polynomial-time evaluation at rational inputs. The normalization to
$[0,1]^2$ only adds multiplications by rational constants and preserves
this property.}}

{
To obtain an instance on the domain required by the formal problem
definition, apply \Cref{rmk:scaling} and output
\[
  \widetilde f(z)=\frac{f(Nz)}{N^3},
  \qquad \alpha=\frac{0.01}{N^2},
  \qquad \widetilde L=\frac{L}{N^2},
  \qquad \widetilde M=\frac{M}{N},
\]
where $L$ and $M$ are the Lipschitz and smoothness bounds for $f$.
An evaluator for $\widetilde f$ and its gradient first computes $Nz$,
uses the evaluators constructed above, and rescales their outputs by
$N^{-3}$ and $N^{-2}$, respectively. These additional operations are
multiplications by rational constants. Since {$\log N=O(n+m)$}, the
constants $N$, $N^{-2}$, $N^{-3}$, and $\alpha$ have {$O(n+m)$}-bit
binary rational descriptions. The rescaled bounds $\widetilde L$ and
$\widetilde M$ also have polynomial encoding length. Thus the
normalized instance and its evaluators can be constructed in polynomial
time and have polynomial description length, despite the exponentially
small value of $\alpha$.

If $z\in[0,1]^2$ is a solution to this instance, then
\Cref{rmk:scaling} yields
\[
  \|Nz-F_{\mathrm{GDA},f}(Nz)\|_2
  \leq N^2\|z-F_{\mathrm{GDA},\widetilde f}(z)\|_2
  <0.01.
\]
The correctness analysis in \Cref{sec:valid} therefore allows us to recover
a solution of $I^{\mathrm{EOPL}}$ from $Nz$ in polynomial time.
Multiplication of the returned coordinates by $N$ increases their
encoding lengths by at most $O(\log N)$. This completes the
polynomial-time reduction to $\textsc{2D-GDAFixedPoint}$ on $[0,1]^2$
and proves \Cref{thm:mainCla}.
}
\end{proof}

\section{Query Lower Bound for Finding Fixed Points of GDA}
\label{sec:query}

An immediate corollary of the above construction is that the set of monotone paths with source $\<0,0\>$ (i.e., those who have source the given input source of EOPL and only solution at a sink) forms a class of hard functions with the following properties:
\begin{enumerate}
    \item Each function (i.e., monotone path) is defined on $[0,N] \times [0,N]$ and has an $\eps$-approximate fixed point of GDA only at the sink sub-square for $\epsilon = 0.01$.

    \item Each function has smoothness parameter $M \leq 100N$.
\end{enumerate}

Our goal in this section is to prove that for this class of functions, any deterministic algorithm with black box query access to the function and gradient oracles needs at least $\Omega(N)$ queries to find a $0.01$-approximate fixed point of GDA when the strategy space of both the min and the max players is $[0,N].$ 

\begin{theorem}
    For any $N > 145$ and any deterministic algorithm $A$ with access to function and gradient evaluation oracles, there is function $f = f_A : [0,N]^2 \to \mathbb R$ that is $100 N$-smooth for which $A$ has to make, in the worst case, $\Omega(N)$ queries in order to find an $0.01$-approximate fixed point of GDA.  
\end{theorem}

The query lower bound follows by a similar path-chasing argument, going back to the work of \cite{hirsch1989exponential,vavasis1993black}. We imagine an algorithm $A$ that makes function and gradient
evaluations. We want to find a pair of functions $f(x,y)$, $f'(x,y)$ on $[0,N]^2$ in our family of hard instances with
disjoint sets of approximate fixed points of $\mathsf{GDA}$ such that algorithm $A$ cannot distinguish
them until it makes a large number of function/gradient queries, i.e., until $\Omega(N)$ queries have been made. Notice
that, by construction, the only approximate fixed points for functions in our family occur in the sink sub-square when $\epsilon = 0.01$, and therefore $f$ and $f'$ will have different sinks. The design of $f,f'$ is algorithm-dependent (based on the queries of $A$) and follows the blueprint of \cite{vavasis1993black,hirsch1989exponential}. The idea is that we create the monotone paths of $f$ and $f'$ adaptively as we see the query points of $A$ and that these two monotone paths will agree at all the query points until the number of function/gradient evaluations gets large enough.

\begin{proof}
Our proof follows exactly the steps of the path-chasing argument of \cite{vavasis1993black}. Consider the partition of $[0,N]^2$ into the sub-squares we have already defined in the previous section.

We assume that algorithm $A$ makes a deterministic sequence of queries, and that at each
 query it evaluates $f$ and $\nabla f$. 
 The sequence of the two-dimensional (adaptively chosen) queries is denoted by $(x_i,y_i) \in [0,N]^2$. 
 Recall that $A$ has unbounded
computational power to pick each query. We will show that the length of the query sequence is at least $\Omega(N)$ in the worst case. First, instead of using the aforementioned query oracle, we will work with the following \emph{local oracle}: a query to this oracle reveals the gadget in one tile, namely all its grid point colors and arrows. We remark that a function/gradient evaluation can be simulated with at most
$4$ such queries (hence this oracle does not affect our lower bound). To see this, choose an elementary unit square containing the evaluation point, and query
the tiles containing its four corners. The absolute function values at these corners depend only on
their coordinates and colors, their gradients depend only on their arrows, and their prescribed mixed
derivatives are zero. These values determine the bicubic interpolant and its gradient on the unit square (at a boundary point, use any incident unit square, with a fixed tie-breaking rule; this is fine since our functions are $C^1$).

Since the function lives at $[0,N]^2$ and each tile that we query is $6 \times 6$, let us set $K = (N+1)/6$ the number of tiles.

We will create adaptively two monotone paths $f, f'$ starting from the bottom left position of the two-dimensional cube. The two monotone paths are  built up as the limit of a sequence $f_0,f_1,f_2,...$ of monotone paths of increasing length with $f_i \subseteq f_{i+1}$ and $f_0 = \<0,0\> \rightarrow \<0,1\> \to \<1,1\>.$ (The last tile $\<1,1\>$ contains a right-up turn, so its outgoing direction is north.)

Let us now fix some round $i$, where we have already constructed the path $f_i$ with last sub-square $\<u_i, v_i\>$.
In our construction, $f$ and $f'$ will both agree with $f_i$ all the way up to sub-square $\<u_i, v_i\>$, and
moreover, that $f$, $f'$ will make a bend in sub-square $\<u_i, v_i\>$. We let $T_i$ denote the
track starting strictly after $\<u_i, v_i\>$ following the direction of the bend and proceeding to the
border of $[0,N]^2$. 
This means that
\begin{enumerate}
    \item For a north-going track, $T_i = \{\<u_i,b\> : K > b > v_i\}$.
    \item For an east-going track, $T_i = \{\<a,v_i\> : K > a > u_i\}$.
\end{enumerate}
The invariant property of the upcoming construction is that no queries have been made in the region of $T_i$ on iterations 1 up to $i$. Since $f_0$ is the path $\<0,0\> \rightarrow \<0,1\>$, we will define by $T_0$ the path that starts from $\<1,1\>$ and moves up until reaching the top boundary. Note that $T_0 \{\<1,b\> : K > b > 1\}$.

\begin{figure}[!ht]
    \centering
    \includegraphics[width=1\linewidth]{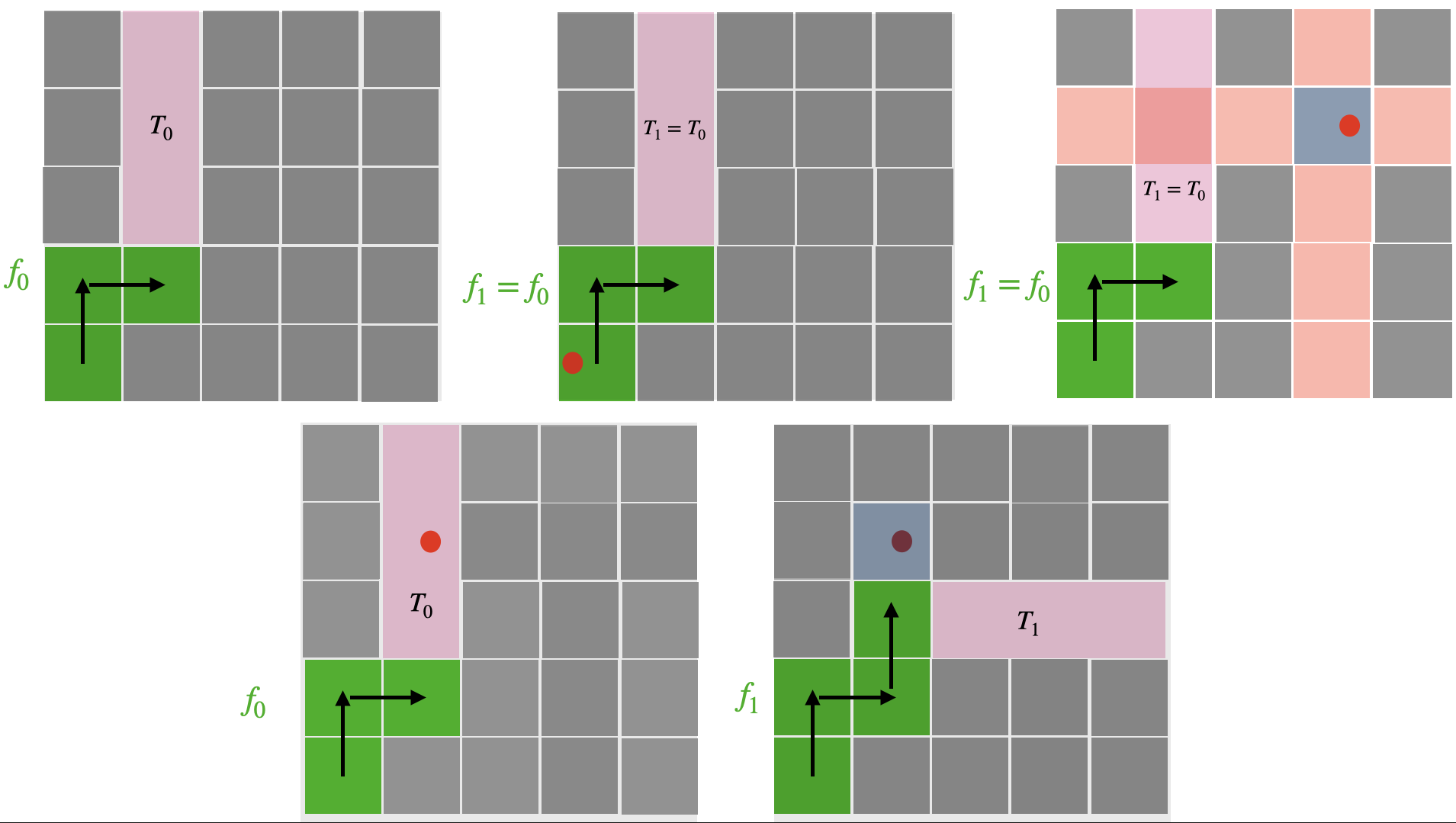}
    \caption{Illustration of the 3 cases that a query (red point) affects the monotone path construction. In the top left figure, we have the original instance which starts with $f_0$ ending at $\<1,1\>$ and $T_0$ be its ``vertical'' extension. In the top middle figure, the query (red) falls in the known part of the path (Case I). Then we make no change to the path and $T_i$. In the top right, the query falls in a potential position of the path. We reveal the background gadget in this square, forbid its row and column as \emph{turning coordinates} (light red) and keep the path unchanged. \emph{Note that the light-red shading does not forbid passing through these rows
and columns (only turning)}. 
    Finally, in the bottom left plot, the query falls in $T_0.$ Here, we extend the path to the first allowed turning row $T_0$ and bend $T_1$ to the right. Finally, we reveal the local tile information at the queried square. }
    \label{fig:placeholder1}
\end{figure}

We also maintain sets $C_i,H_i$ of \emph{forbidden turning columns and rows}, which are both initially empty sets. A column in
$C_i$ cannot contain an east-to-north turn, and a row in $H_i$ cannot contain a north-to-east turn. However, it is important to note that passing
through these rows and columns is allowed.

For a current tile $\<u_i,v_i\>$, we will study the upper east-north region $R_i = \{\<a,b\> : u_i < a < K, v_i < b < K \}.$

In order to extend $f_{i-1}$ to $f_i$ for $i = 1,2,...$ we consider 3 cases. To set up some notation, we have already defined $f_{i-1}$, ending at sub-square $\<u_{i-1}, v_{i-1}\>$ and let $T_{i-1}$ be its extension to the boundary after the ``bending'' that we discussed before. Then the query $w_i$ is made to the local oracle:
\begin{enumerate}
    \item Case I: The query $w_i$ lies in the part of the space that is already determined/revealed to the algorithm. 
    This case includes a repeated query (we give the same answer as before), a tile in $f_{i-1}$ (we give the path gadget of the tile), or a tile outside
the region $f_{i-1} \cup T_{i-1} \cup R_{i-1}$ (we give the background gadget).
    Then we set $f_i = f_{i-1}$ (the algorithm spent a query for free). We also set $T_{i} = T_{i-1}$ and $\<u_i, v_i\> = \<u_{i-1}, v_{i-1}\>$.

    \item Case II: The query $w_i$ lies in a sub-square $\<u',v'\>$ of $[0,N]^2$, not in $T_{i-1}$ (that is strictly northeast
of the last tile of the prefix). This means that $u' > u_{i-1}$ and $v' > v_{i-1}.$ 
Then we forbid column $u'$
for east-to-north turns and row $v'$ for north-to-east turns, setting $C_i = C_{i-1} \cup \{u'\}$
and $H_i = H_{i-1} \cup \{v'\}.$
    We then keep $f_{i} = f_{i-1}$, $T_{i} = T_{i-1}$, and $\<u_i, v_i\> = \<u_{i-1}, v_{i-1}\>$.
    We also set the queried sub-square as the background tile, telling the algorithm that our path does not pass through it.

    \item Case III:
    In the third case, $w_i$ lies in $T_{i-1}$. In this case (assuming for example that
it is a north track), we let $v_i$ be the smallest integer coordinate greater than $v_{i-1}$ (current column)
whose row is not forbidden for
turning. Equivalently,
\[
v_i = \min \{v > v_{i-1} : v \in \mathbb Z, v \notin H_{i-1}\}.
\]
A counting argument can show that this coordinate lies inside the grid throughout the rounds
under consideration.
Then
we let $f_i$ be the union of the path $f_{i-1}$ and the portion of $T_{i-1}$ connecting 
$\<u_{i-1}, v_{i-1}\> \rightarrow \<u_{i-1}, v_i\>$. We also let $u_i = u_{i-1}$ and $T_i$ be the track moving to the right (it bends the previous path) starting at $\<u_i, v_i\>$ and
including sub-squares with increasing first coordinates. 
We fix the up-right gadget at the new bend. If $w_i$ lies below the bend, we return the
straight north-going gadget; if it is the bend, we return the up-right gadget; and if it lies above the
bend, we return background. In the last case, every subsequent continuation has first coordinate
strictly larger than $u_i$, so the queried tile remains background. If $T_{i-1}$ had been
a straight track moving right, then $T_i$ would have been a straight track moving up. The other case is treated symmetrically.
\end{enumerate}

First, one can verify that all the invariants of the construction hold after the above choices.

We will conclude the proof with an
indistinguishability argument showing that after few queries, we can still construct two functions/paths that agree with every answer given to the algorithm, but have different solutions. We notice that if $f_i$ terminates at
sub-square $\<u_i,v_i\>$, then $u_i, v_i \leq i+1$ because $f_0$ starts with $u_0=v_0 = 1$ and we never pass to a higher value of
$u$ unless all lower integer values of $u$ had queries associated with them. The same
holds for the second coordinate. Therefore, we can continue extending $f_i$ until $\Omega(N)$ queries have been made in the domain $[0,N]^2$.

In particular, after any $q \leq K-7$ local queries, we can still complete the prefix in two different ways.
First, note that
the sink gadget can only be entered from below, so we first have to ensure that the final track is north-going. This can happen to a small adaptation: if $T_q$ already points north, then no change is needed; otherwise we can extend the prefix along $T_q$ to the first column outside $C_q$ and
make a right-up turn there. This added segment contains no previous query and the resulting north-going
track is also not queried. If $\<u,v\>$ is the last bend
after this potential extension, then we know that $u \leq q+2, v \leq q+1$ (by the same counting as before). Now we can complete the two constructions: one path will stop at $\<u,v+2\>$, and the other will stop at $\<u,v+4\>$, using the
sink gadget in each case. Note that both completions stay inside the grid and leave the last tile row and column
as background. They change only tiles that are not queried and their sink sub-squares are disjoint. Consequently, the two completed functions agree at every simulated function/gradient
query and have disjoint sets of $0.01$-approximate GDA fixed points. This establishes the query lower bound.
\end{proof}

\subsection*{AI Disclosure}

All the conceptual and technical contributions of the paper were obtained by the authors by the end of December 2025 without the assistance of AI models.
After writing a full version of the paper, AI models were only used to help streamline the writing, find typographical errors, and verify some of the proofs.

\bibliography{bib}

\appendix

{
\section{Proof of Lemma~\ref{lemma:sinks-boundary}}
\label{app:cases}
The goal of this appendix is to establish \Cref{lemma:sinks-boundary}.
Throughout, let $\epsilon=0.01$.

\begin{lemma}
[Bad points are not approximate fixed points of $\mathsf{GDA}$]
Given an EOPL instance, let $f:[0,N]^2\to\mathbb R$ be the associated
function of \Cref{sec:extension}. If $(x,y)$ is a bad point of $f$,
then it is not an $\epsilon$-approximate fixed point of $\mathsf{GDA}$.
\end{lemma}

The outer layer contains eight types of small squares, which we consider
separately below. Within each small square, $(x,y)\in[0,1]^2$ denotes
local coordinates measured from its bottom-left corner. We use $f$ for
the restriction of the constructed function to this square, expressed
in these coordinates. In particular, $(0,1)$ is the top-left corner and
$(1,0)$ is the bottom-right corner. Projection is always onto the full
domain $[0,N]^2$, rather than onto an individual small square.

Good points in these configurations are excluded by the group analysis
in \Cref{sec:proof1}, so it remains to exclude bad points. A bad
$\epsilon$-approximate fixed point must lie within distance $\epsilon$
of a domain boundary at which projection is active. Types $2$--$5$ and
$7$ meet only the left boundary, so for these cases it suffices to
consider $x\le\epsilon$. The corner cases for Types $0$ and $6$ are
addressed below; Type $1$ is handled on its entire square.

\subsection{Type 0 squares}
\begin{lemma}
A type $0$ square contains no $\epsilon$-approximate fixed point of
$\mathsf{GDA}$.
\end{lemma}

\begin{figure}[!ht]
    \centering
    \includegraphics[width=0.3\linewidth]{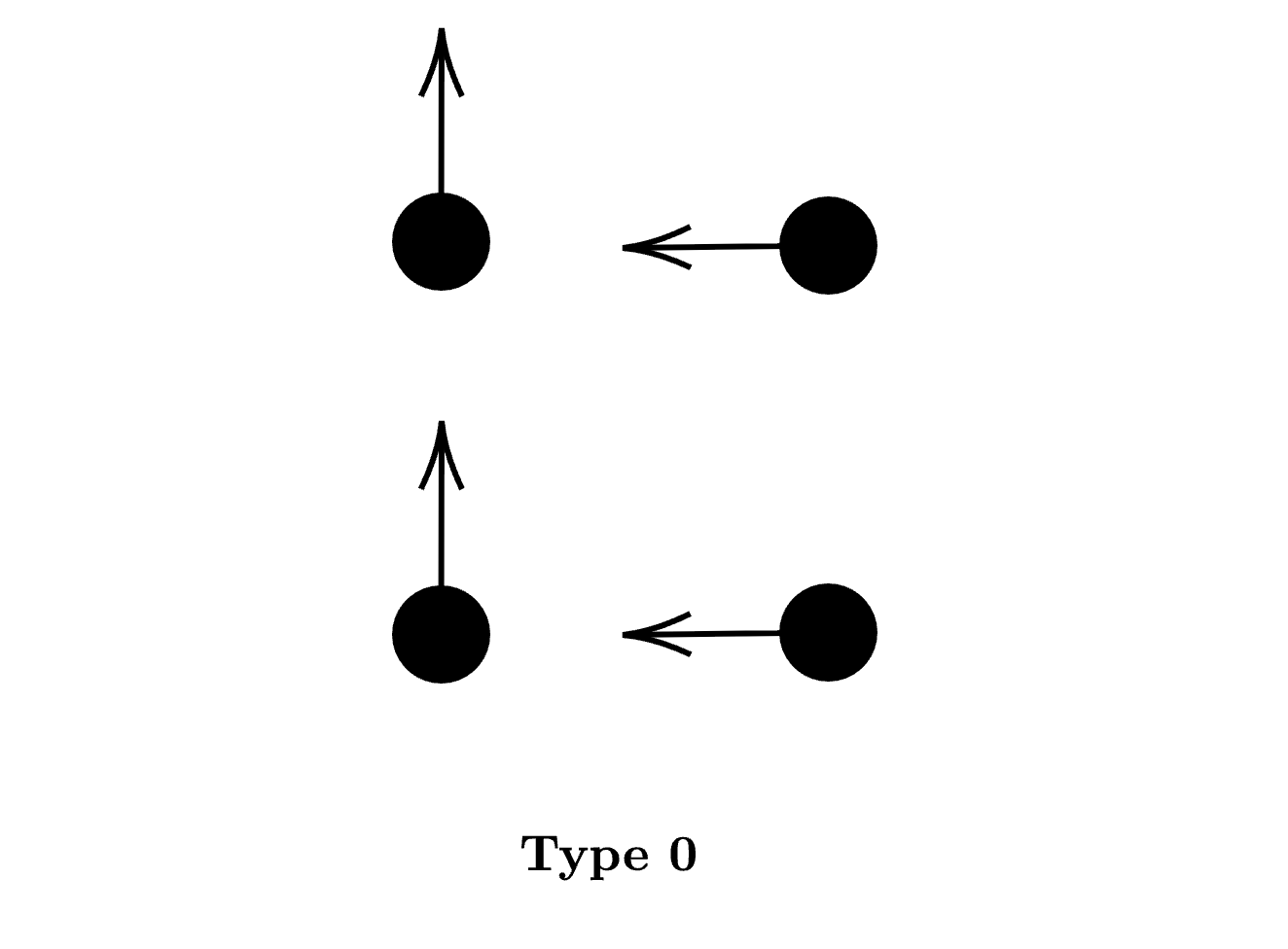}
    \caption{Type 0}
    \label{fig:placeholder2}
\end{figure}

Type $0$ squares, shown in \Cref{fig:placeholder2}, occur along the
left boundary and do not meet the bottom boundary. Bicubic interpolation
gives the following identity.

\begin{lemma}\label{l:1}
For a type $0$ square,
\[
\begin{aligned}
f_y(x,y)={}&-6x^3y^2+6x^3y-x^3+9x^2y^2-9x^2y\\
          &+1.5x^2+3y^2-3y-0.5.
\end{aligned}
\]
\end{lemma}

Writing $h(x)=3x^2-2x^3$, \Cref{l:1} gives
\[
f_y(x,y)=-3(1+h(x))y(1-y)-\frac{1-h(x)}2\le0.
\]
Thus projection at the top boundary cannot be active. Projection at
the bottom boundary would give a displacement of at least $1$.
Consequently, a bad $\epsilon$-approximate fixed point in this square
would require projection at the left boundary and hence $x\le0.01$.
For such $x$, the same identity gives
\[
f_y(x,y)\le-\frac{1-h(0.01)}2=-0.499851<-0.49.
\]
Since the square does not meet the bottom boundary, the projected
vertical displacement is at least $0.49>\epsilon$, a contradiction.

\subsection{Type 1 squares}
\begin{lemma}
A type $1$ square contains no $\epsilon$-approximate fixed point of
$\mathsf{GDA}$.
\end{lemma}

Type $1$ squares, shown in \Cref{fig:type1}, occur away from the left
boundary, including on the bottom, top, and right boundaries. They
also occur in the interior.

\begin{figure}[!ht]
    \centering
    \includegraphics[width=0.3\linewidth]{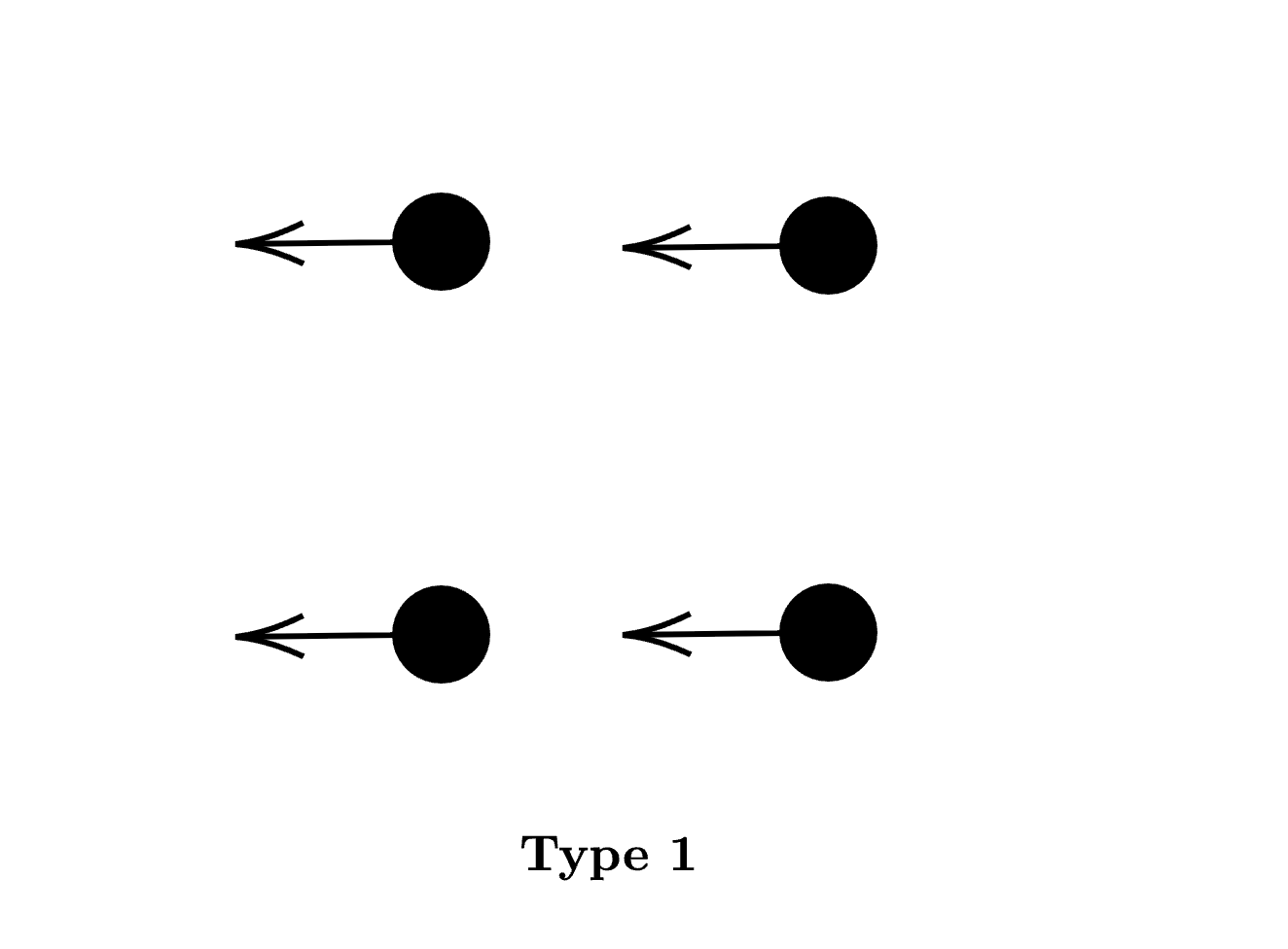}
    \caption{Type 1}
    \label{fig:type1}
\end{figure}

\begin{lemma}\label{l:type1square}
For a type $1$ square,
\[
f_x(x,y)=-3x^2+3x+0.5.
\]
\end{lemma}

By \Cref{l:type1square}, $f_x(x,y)\ge0.5$. Every point in such a square
has global horizontal coordinate at least $1$, so projection onto the
domain leaves a horizontal displacement of at least $0.5>\epsilon$.

\subsection{Type 2 squares}
\begin{lemma}
A type $2$ square contains no $\epsilon$-approximate fixed point of
$\mathsf{GDA}$.
\end{lemma}

\begin{figure}[!ht]
    \centering
    \includegraphics[width=0.3\linewidth]{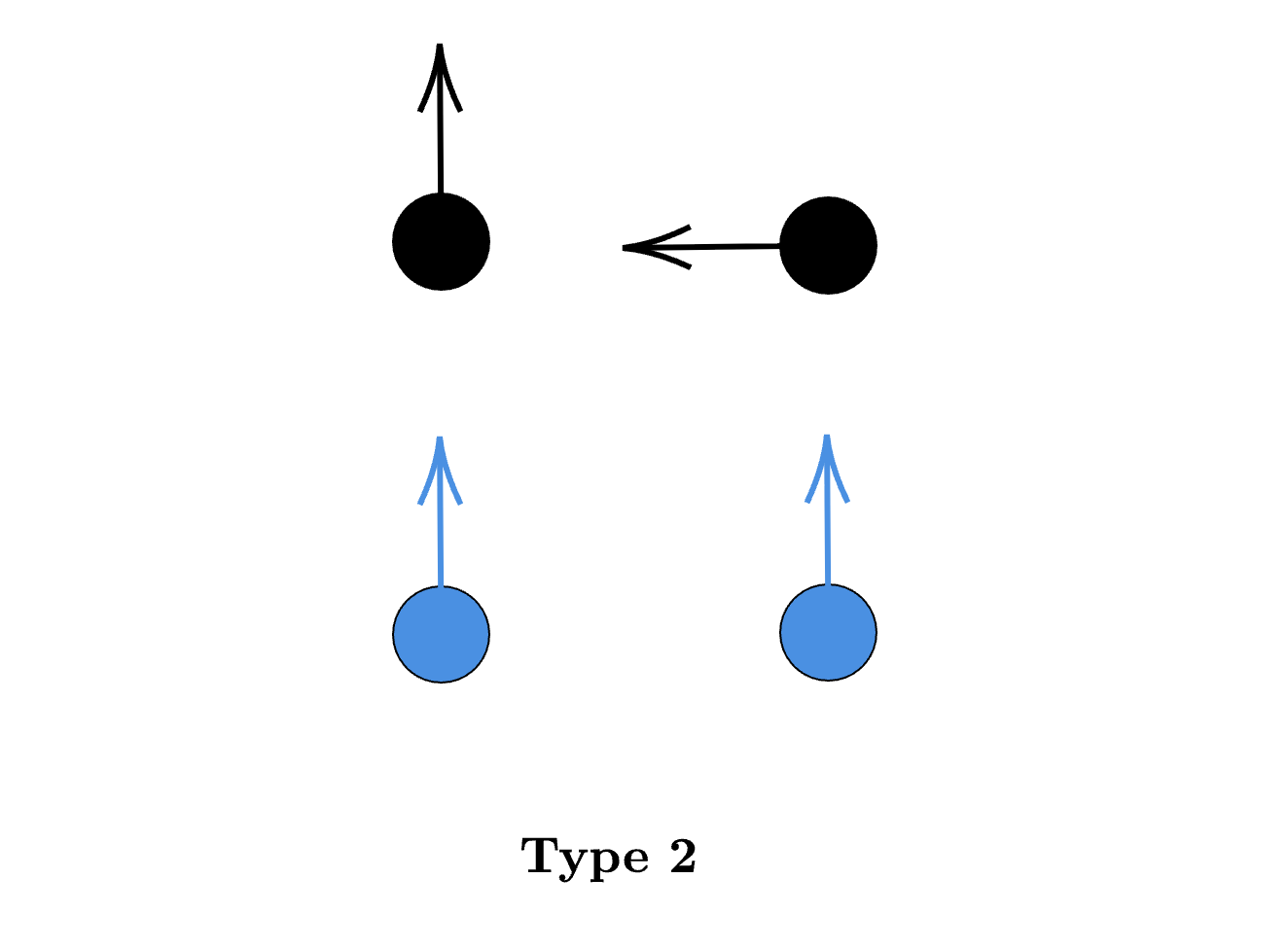}
    \caption{Type 2, Group 3}
    \label{fig:type2}
\end{figure}

Type $2$ squares, shown in \Cref{fig:type2}, meet the left boundary
but not the bottom boundary. It suffices to consider $x\le0.01$.

\begin{lemma}\label{l:type2square}
For a type $2$ square, let $D=f(1,0)-f(0,1)$, the value at the blue
bottom-right corner minus the value at the black top-left corner. Then
\[
\begin{aligned}
f_y(x,y)={}&18x^3y^2-19x^3y-\frac{57}{2}x^2y^2+30x^2y\\
          &-(6D+3)y(1-y)-\frac12.
\end{aligned}
\]
\end{lemma}

The value regimes in \Cref{sec:function-values} give $D\ge0$.
Using \Cref{l:type2square}, for $0\le x\le0.01$ and $0\le y\le1$,
\[
\begin{aligned}
f_y(x,y)
&\le x^3y(18y-19)+30x^2y-\frac12\\
&\le30x^2-\frac12\le-0.497<-0.49.
\end{aligned}
\]
Since the square does not meet the bottom boundary, the projected
vertical displacement is at least $0.49>\epsilon$.

\subsection{Type 3 squares}
\begin{lemma}
A type $3$ square contains no $\epsilon$-approximate fixed point of
$\mathsf{GDA}$.
\end{lemma}

\begin{figure}[!ht]
    \centering
    \includegraphics[width=0.3\linewidth]{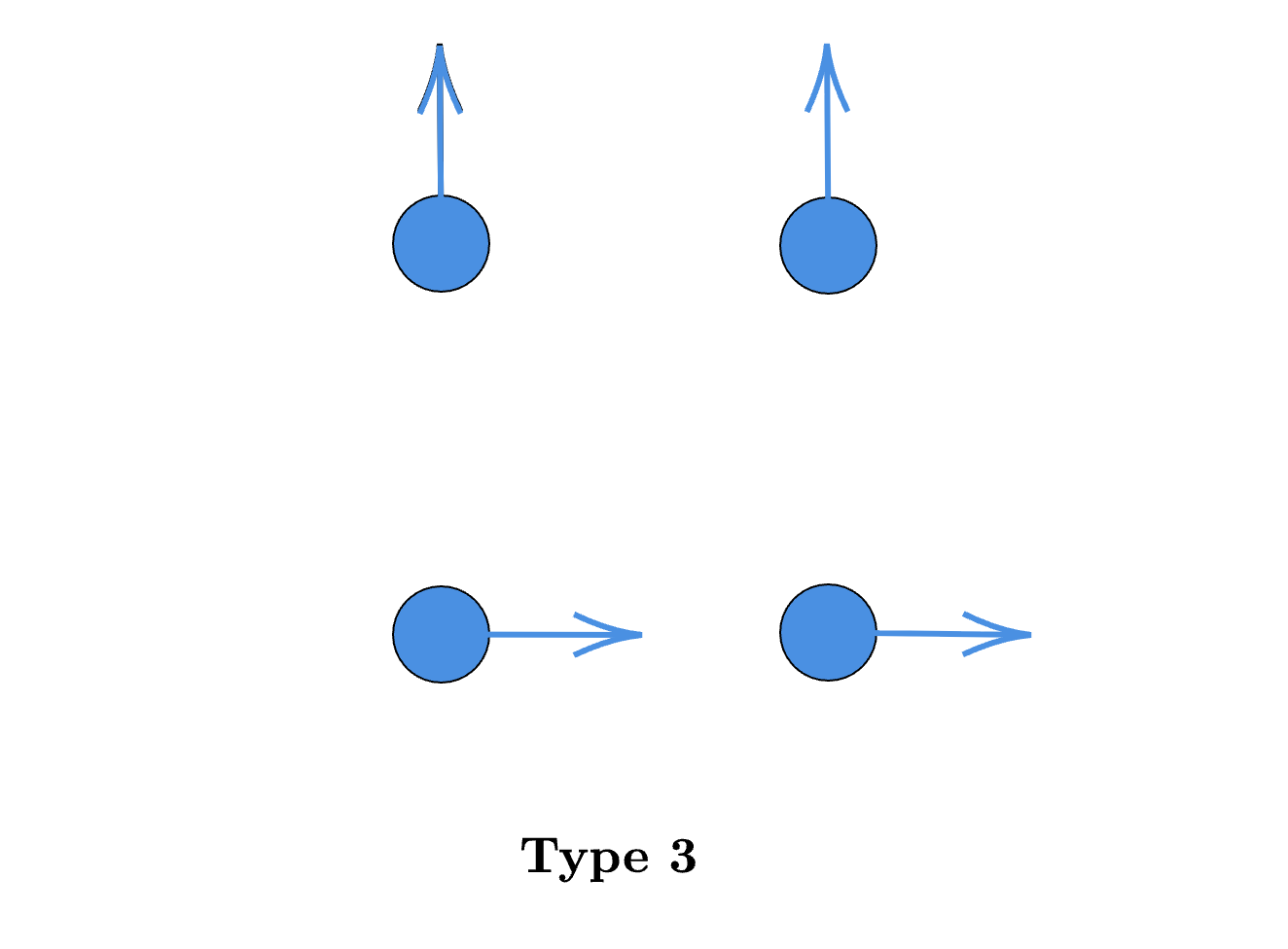}
    \caption{Type 3}
    \label{fig:type3}
\end{figure}

Type $3$ squares, shown in \Cref{fig:type3}, meet the left boundary.
It suffices to consider $x\le0.01$.

\begin{lemma}\label{l:2}
For a type $3$ square,
\[
\begin{aligned}
f_x(x,y)={}&-6x^2y^3+9x^2y^2+3x^2+6xy^3-9xy^2\\
          &-3x-y^3+1.5y^2-0.5,\\
f_y(x,y)={}&y(-6x^3y+6x^3+9x^2y-9x^2-3xy+3x+4.5y-5).
\end{aligned}
\]
\end{lemma}

By \Cref{l:2}, for $x\le0.01$,
\[
f_x(x,y)\le-y^3+1.5y^2-0.45,
\qquad
f_y(x,y)\le y(4.5y-4.95).
\]
If $y\le0.5$, then $f_x(x,y)\le-0.2$ and the rightward displacement
is at least $0.2$, since at least $1-x\ge0.99$ remains to the right
within the square. If $y\ge0.5$, then $f_y(x,y)\le-0.45$ and at least
half a unit remains below within the square, giving a downward
displacement of at least $0.45$. Both bounds exceed $\epsilon$.

\subsection{Type 4 squares}
\begin{lemma}
A type $4$ square contains no $\epsilon$-approximate fixed point of
$\mathsf{GDA}$.
\end{lemma}

\begin{figure}[!ht]
    \centering
    \includegraphics[width=0.3\linewidth]{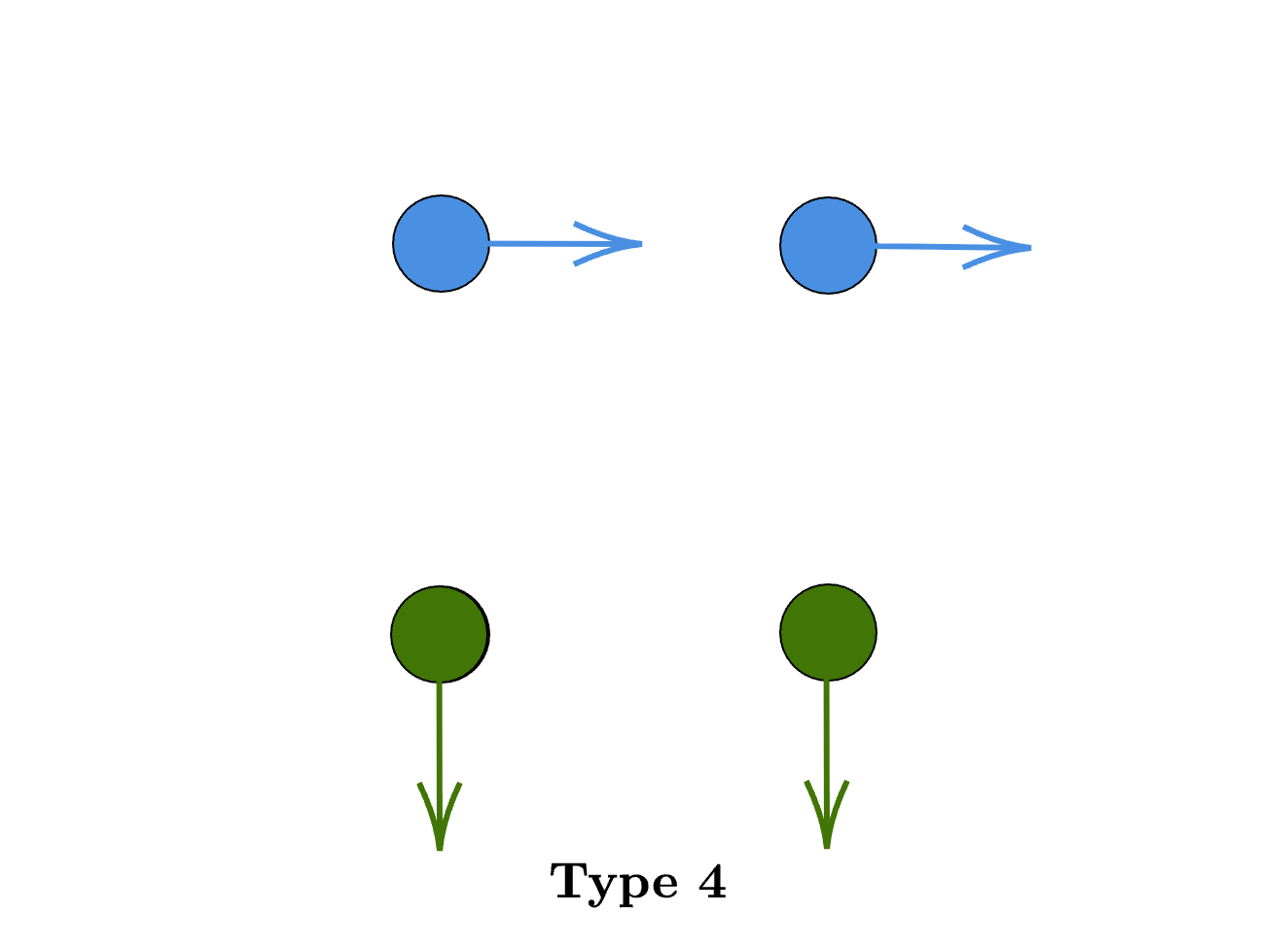}
    \caption{Type 4, Group 2}
    \label{fig:type4}
\end{figure}

Type $4$ squares, shown in \Cref{fig:type4}, meet the left boundary.
It suffices to consider $x\le0.01$.

\begin{lemma}\label{l:type4square}
For a type $4$ square,
\[
\begin{aligned}
f_x(x,y)={}&-18x^2y^3+27x^2y^2-6x^2+18xy^3-27xy^2\\
          &+6x+y^3-1.5y^2,\\
f_y(x,y)={}&-18x^3y^2+18x^3y+27x^2y^2-27x^2y+3xy^2-3xy\\
          &+(6f(1,1)-6f(0,0)+4.5)y(1-y)-0.5y+0.5,
\end{aligned}
\]
where $f(0,0)$ and $f(1,1)$ are the function values at the bottom-left
and top-right corners, respectively.
\end{lemma}

The value regimes in \Cref{sec:function-values} give
$6f(1,1)-6f(0,0)+4.5\ge0$. For $x\le0.01$, \Cref{l:type4square} yields
\[
\begin{aligned}
f_x(x,y)
&=x(1-x)(18y^3-27y^2+6)+y^3-1.5y^2\\
&\le0.06+y^3-1.5y^2,\\
f_y(x,y)&\ge-0.5y+0.49.
\end{aligned}
\]
If $y\le0.5$, then $f_y(x,y)\ge0.24$ and at least half a unit remains
above within the square, giving an upward displacement of at least
$0.24$. If $y\ge0.5$, then $f_x(x,y)\le-0.19$ and at least
$1-x\ge0.99$ remains to the right, giving a rightward displacement
of at least $0.19$. Both bounds exceed $\epsilon$.

\subsection{Type 5 squares}
\begin{lemma}
A type $5$ square contains no $\epsilon$-approximate fixed point of
$\mathsf{GDA}$.
\end{lemma}

\begin{figure}[!ht]
    \centering
    \includegraphics[width=0.3\linewidth]{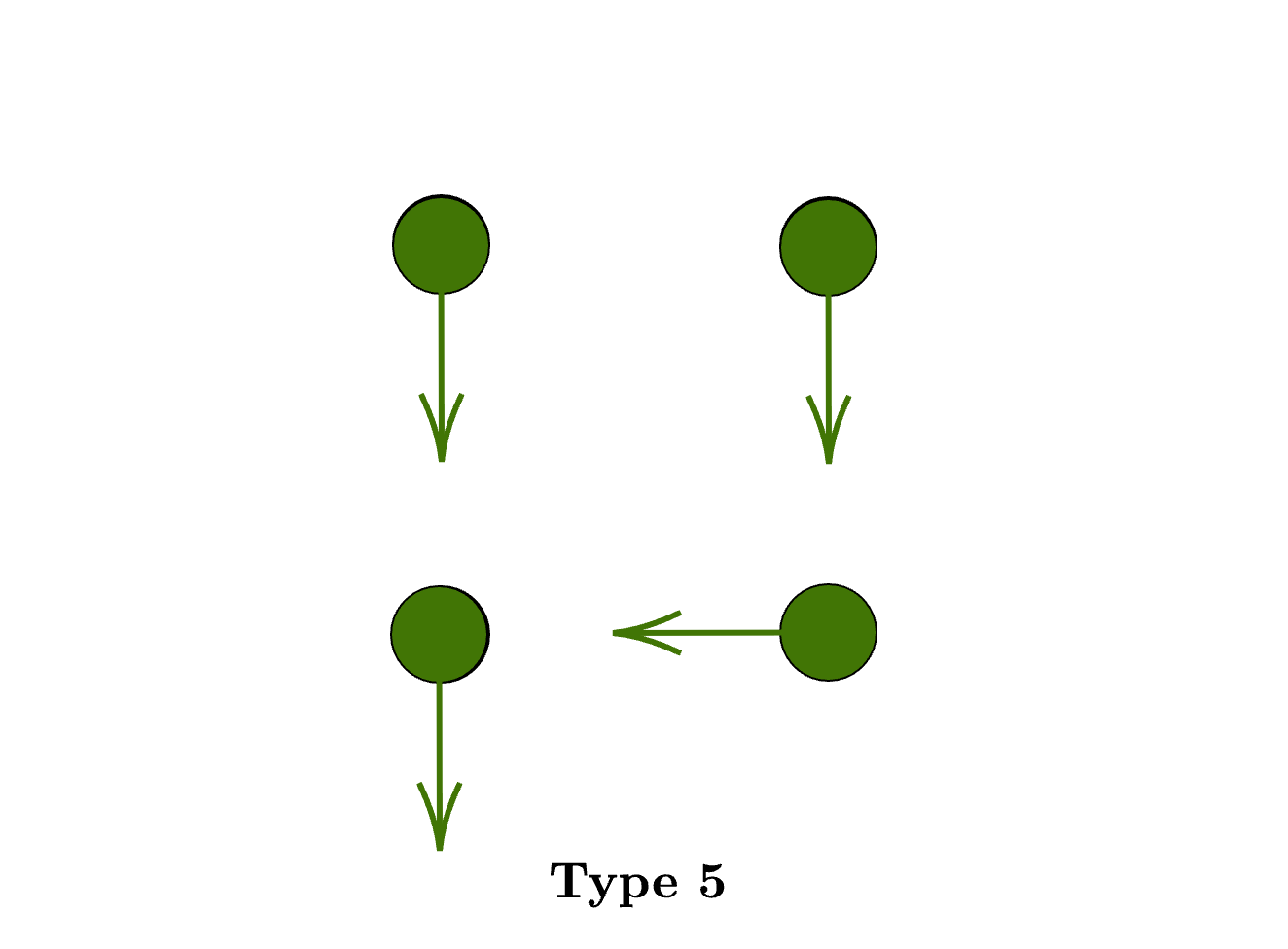}
    \caption{Type 5, Group 3}
    \label{fig:type5}
\end{figure}

Type $5$ squares, shown in \Cref{fig:type5}, meet the left boundary
but not the top boundary. It suffices to consider $x\le0.01$.

\begin{lemma}\label{l:3}
For a type $5$ square,
\[
\begin{aligned}
f_x(x,y)={}&x(6xy^3-10.5xy^2+3xy-4.5x-5y^3+9y^2-3y+5),\\
f_y(x,y)={}&6x^3y^2-7x^3y+x^3-7.5x^2y^2+9x^2y\\
          &-1.5x^2-3y^2+3y+0.5.
\end{aligned}
\]
\end{lemma}

By \Cref{l:3}, for $x\le0.01$,
\[
f_y(x,y)\ge-7x^3y-7.5x^2y^2-1.5x^2+0.5\ge0.45.
\]
Since the square does not meet the top boundary, the projected
vertical displacement is at least $0.45>\epsilon$.

\subsection{Type 6 squares}
\begin{lemma}
A type $6$ square contains no $\epsilon$-approximate fixed point of
$\mathsf{GDA}$.
\end{lemma}

\begin{figure}[!ht]
    \centering
    \includegraphics[width=0.3\linewidth]{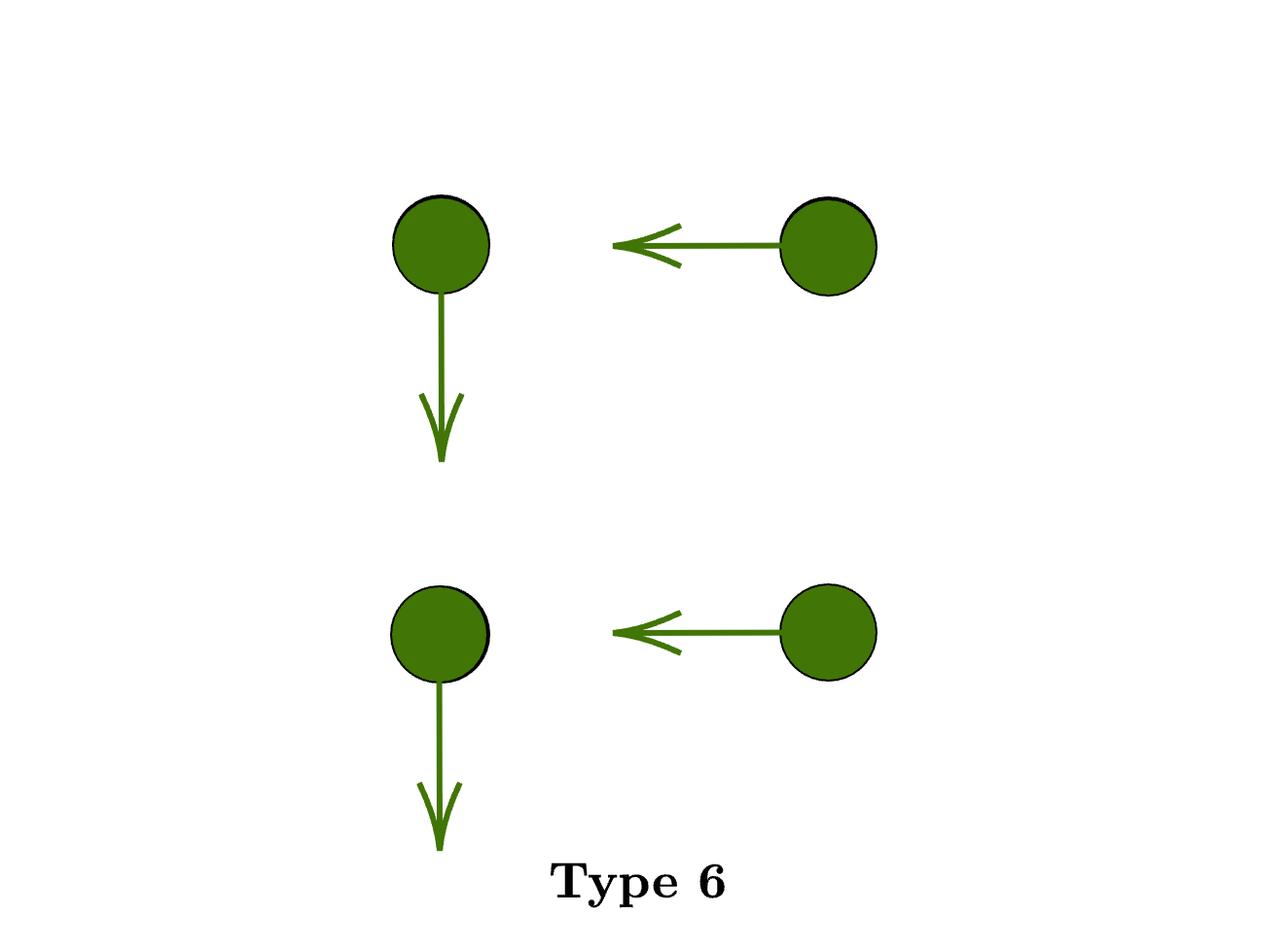}
    \caption{Type 6, Group 2}
    \label{fig:type6}
\end{figure}

Type $6$ squares, shown in \Cref{fig:type6}, meet the left boundary
and may also meet the bottom boundary, but do not meet the top boundary.

\begin{lemma}\label{l:4}
For a type $6$ square,
\[
\begin{aligned}
f_x(x,y)={}&x(6xy^3-9xy^2+3xy-4.5x-6y^3+9y^2-3y+5),\\
f_y(x,y)={}&6x^3y^2-6x^3y+x^3-9x^2y^2+9x^2y\\
          &-1.5x^2-3y^2+3y+0.5.
\end{aligned}
\]
\end{lemma}

Writing $h(x)=3x^2-2x^3$, \Cref{l:4} gives
\[
f_y(x,y)=3(1+h(x))y(1-y)+\frac{1-h(x)}2\ge0.
\]
Thus projection at the bottom boundary cannot be active, while
projection at the top boundary would give a displacement of at least
$1$. A bad $\epsilon$-approximate fixed point would therefore require
projection at the left boundary, forcing $x\le0.01$. For these points,
\[
f_y(x,y)\ge-6x^3y-9x^2y^2-1.5x^2+0.5\ge0.45.
\]
Since the square does not meet the top boundary, the projected
vertical displacement is at least $0.45>\epsilon$.

\subsection{Type 7 squares}
\begin{lemma}
A type $7$ square contains no $\epsilon$-approximate fixed point of
$\mathsf{GDA}$.
\end{lemma}

\begin{figure}[!ht]
    \centering
    \includegraphics[width=0.3\linewidth]{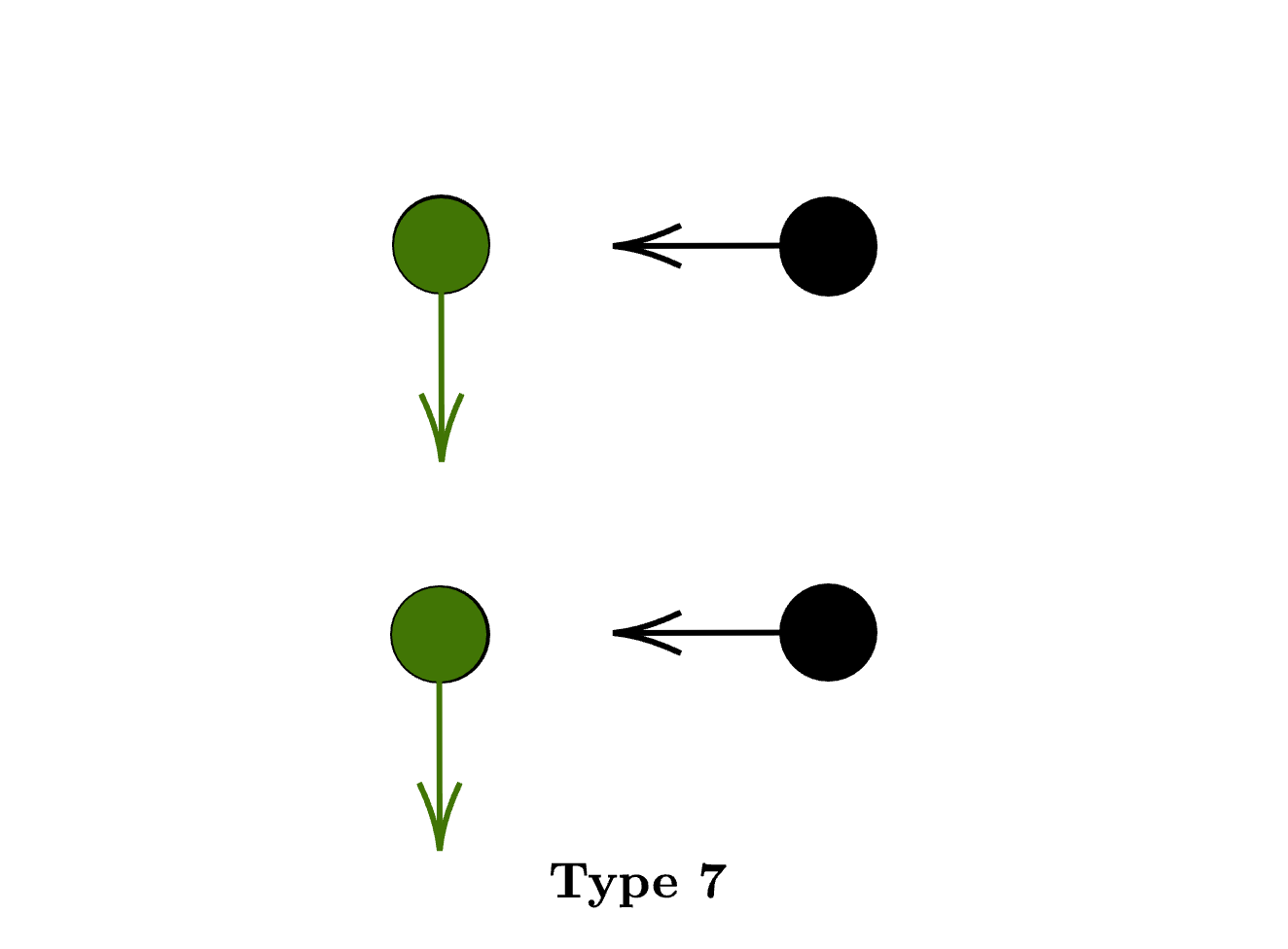}
    \caption{Type 7, Group 2}
    \label{fig:type7}
\end{figure}

Type $7$ squares, shown in \Cref{fig:type7}, meet the left boundary
but not the top boundary. It suffices to consider $x\le0.01$.

\begin{lemma}\label{l:5}
For a type $7$ square, write $b=f(0,0)$ and $c=f(1,0)$ for the values
at the green bottom-left and black bottom-right corners, respectively.
Then $f(0,1)=b+1$, $f(1,1)=c-1$, and
\[
\begin{aligned}
f_x(x,y)={}&x\bigl(-18xy^3+27xy^2+3xy+x(6b-6c+1.5)\\
          &\hspace{2.2em}+18y^3-27y^2-3y-6b+6c-1\bigr),\\
f_y(x,y)={}&-18x^3y^2+18x^3y+x^3+27x^2y^2-27x^2y\\
          &-1.5x^2-3y^2+3y+0.5.
\end{aligned}
\]
\end{lemma}

By \Cref{l:5}, for $x\le0.01$,
\[
f_y(x,y)\ge-18x^3y^2-27x^2y-1.5x^2+0.5\ge0.45.
\]
Since the square does not meet the top boundary, the projected
vertical displacement is at least $0.45>\epsilon$.

The eight cases exclude bad $\epsilon$-approximate fixed points
throughout the outer layer and complete the proof of
\Cref{lemma:sinks-boundary}.
{
\begin{remark}
    We remark that the the unit square $[1,2] \times [0,1]$ of Figure 4(a) has
green left corners, black right corners, and four left-pointing arrows. This does not belong to one of the above types but it is easy to check directly that it does not contain an approximate GDA fixed point.
\end{remark}
}
\newpage
}

\end{document}